\documentclass[11pt]{article}
\usepackage{graphicx,subfigure}
\usepackage{amsfonts}
\usepackage{amssymb}
\usepackage{diagbox} 
\usepackage{float}
\usepackage{rotating}
\usepackage{eqparbox}
\usepackage{makecell,caption}
\usepackage[round]{natbib}
\usepackage{amsthm}
\usepackage{amsmath,amsfonts,amssymb,amsthm}
\usepackage[utf8]{inputenc}
\usepackage{hyperref}
\usepackage{multirow}
\usepackage{multicol}
\usepackage{amsthm}
\usepackage{array}
\usepackage{color}
\usepackage{lscape}
\usepackage{booktabs}
\usepackage{multirow}

\usepackage{booktabs,tabularx}
\newcolumntype{x}{>{\raggedright\arraybackslash}X}
\makeatletter
\newcommand{\doublespacing}{\let\CS=\@currsize\renewcommand{\baselinestretch}{1.05}\tiny\CS}
\begin{document}
	\newtheorem{corollary}{Corollary}[section]
	\newtheorem{thm}{Theorem}[section]
	\newtheorem{lemma}{Lemma}[section]
	\newtheorem{remark}{Remark}[section]
	\newtheorem{definition}{Definition}[section]
	\newtheorem{Example}{Example}[section]
	\newtheorem{Remark}{Remark}[section]
	\numberwithin{equation}{section}
	
	\newtheorem{notation}{Notation}[section]
	\newtheorem{proposition}{Proposition}[section]

	\renewcommand{\theequation}{\thesection.\arabic{equation}}
	\doublespacing
	\vspace{5cm}
	\title{\bf Weighted cumulative past inaccuracy and Kullback-Leibler divergence based on extropy: properties, estimation, and applications}
	\author{
		\large Bighneswar \textbf{Sahoo}\thanks{E-mail: bsahoo7500@gmail.com}~ and~
		\large Suchandan \textbf{Kayal}\thanks{Corresponding author. E-mail: kayals@nitrkl.ac.in, suchandan.kayal@gmail.com} 
	}
	\date{}
	\maketitle
	\noindent{\it Department of Mathematics, National Institute of Technology Rourkela, Rourkela-769008, Odisha, India.}
	\vspace*{.5cm}
	
	\begin{center}
		\noindent{\bf Abstract}	
	\end{center}
		This study develops a weighted framework for measuring the discrepancy between two nonnegative lifetime distributions through cumulative past extropy. We propose two measures, referred
		to as the weighted cumulative past extropy inaccuracy (WCPEI) and the weighted cumulative past extropy Kullback-Leibler divergence (WCPED). The generalized weight function is considered in this study.  We investigate a number of theoretical properties of these measures. Empirical distribution function-based nonparametric estimator is subsequently constructed for the weighted cumulative past extropy inaccuracy ration (WCPEIR). Its finite-sample behavior is studied through Monte Carlo simulation experiments for different sample sizes. To illustrate the practical relevance of the WCPED, two applications are considered. First, an extropy-based goodness-of-fit procedure for
		testing uniformity is developed using the proposed divergence measure. Its power is then compared with that of several established uniformity tests under a variety of alternatives. Second, an image analysis application is presented in which the proposed measure is employed to assess changes in the distributions of pixel intensities
		when the image resolution is altered. The framework is further extended to a dynamic setting by conditioning
		on the lifetime information available up to a specified time point. This leads to the dynamic weighted cumulative past extropy inaccuracy (DWCPEI) and dynamic weighted cumulative past extropy divergence (DWCPED). Their theoretical properties are derived, and the
		corresponding nonparametric estimation procedures are proposed. The finite-sample performance of these estimators is evaluated through simulation studies using $R$ software.

%
%

\section{Introduction}\label{sec1}
Inaccuracy measures play a fundamental role in information theory and statistics by quantifying the discrepancy between two probability distributions. Unlike classical uncertainty measures, which describe the randomness of a single distribution, inaccuracy measures evaluate the loss of information incurred when one probability model is used to represent another. The pioneering work of \cite{kerridge1961inaccuracy} introduced the concept of inaccuracy based on the probability density functions (PDFs), establishing an important connection between entropy and discrimination information. Let $X_1$ and $X_2$ be two absolutely continuous nonnegative random
variables (RVs) with cumulative distribution functions (CDFs) $F_1(\cdot)$
and $F_2(\cdot)$, and PDFs $f(\cdot)$ and $g(\cdot)$, respectively. The classical Kerridge
inaccuracy measure is then defined by
\begin{equation}
	I(X_1,X_2)
	=
	-\int_0^\infty f(y)\log g(y)\,dy= S(X_1) + K(X_1, X_2),
\end{equation}
where the Shannon entropy (see \cite{shannon1948mathematical}) and the Kullback Leibler (KL) divergence (see \cite{kullback1951information}) are respectively given by
\begin{equation}
S(X_1) = -\int_0^\infty f(y) \log f(y) \, dy~~\mbox{and}~~K(X_1, X_2) = \int_0^\infty f(y) \log \frac{f(y)}{g(y)} \, dy.
\end{equation}
A fundamental property of the inaccuracy measure is that it reduces to the Shannon's uncertainty measure when the two distributions are identical.

In recent years, extropy has emerged as a complementary measure of
uncertainty to the Shannon entropy. Following \cite{lad2015extropy}, the extropy of  $X_1$, with PDF $f(\cdot)$, is defined as
\begin{equation}
	J(X_1)=-\frac{1}{2}\int_{0}^{\infty} f^2(y)\,dy.
\end{equation}
This measure captures the ``spread'' or ``dispersion'' of a distribution in a manner analogous to entropy but with distinct theoretical properties. The extropy of a distribution attains its maximum for the uniform distribution, contrasting with entropy which attains its maximum under the same condition. \cite{qiu2017extropy} characterized the extropy of order statistics and record values, while \cite{qiu2018extropy} provided estimators for the extropy of a continuous RV.
Building upon the extropy framework, \cite{kundu2023cumulative} introduced cumulative past extropy  based on the CDF  ${F_{1}}(\cdot)$ and its dynamic version as
\begin{equation}\label{eq1.4}
	J_p(X_1) = -\frac{1}{2} \int_0^\infty {F_{1}}^2(y) \, dy~\mbox{and}~J_p(X_1; t) = -\frac{1}{2} \int_0^t \left( \frac{{F_{1}}(y)}{F_{1}(t)} \right)^2 \, dy,
\end{equation}
respectively.
The CDF-
based measures offer the advantage of being well-defined even in cases where the PDF may not exist, making them more versatile in practical applications.
In parallel developments, \cite{hashempour2024new} proposed extropy-based inaccuracy measure for record statistics, defined as:
\begin{equation}\label{eq1.8}
	J(f, g) = -\frac{1}{2} \int_0^\infty f(y)g(y) \, dy
\end{equation}
and the extropy-based discrimination information:
\begin{equation}\label{eq1.9}
	J(f|g) = \frac{1}{2} \int_0^\infty [f(y) - g(y)]f(y) \, dy.
\end{equation}
The development of cumulative information measures has attracted considerable attention because CDFs are generally smoother and more robust than the PDF. Consequently, cumulative entropy, cumulative residual entropy, cumulative extropy, and cumulative past extropy have emerged as useful alternatives to their density-based counterparts. Among these, the cumulative past extropy has proved particularly valuable for analyzing the uncertainty associated with elapsed lifetimes, making it naturally suited for applications involving left-censored observations, repairable systems, maintenance planning, medical follow-up studies, and reliability investigations where the past lifetime information is of primary interest. Recently, the cumulative inaccuracy and divergence measures based on extropy have been introduced and studied in the literature. In this direction, we refer to \cite{hashempour2024dynamic}, \cite{saranya2024relative},  \cite{saranya2025inaccuracy}, \cite{hashempour2026dynamic}, and \cite{saranya2026relative}. Dynamic versions have also been developed by conditioning on survival or failure before a specified time point, thereby enabling localized comparisons of distributions over time. Such conditional measures are especially useful in reliability and survival analysis, where the behavior of components often changes with age. 

Compared with the extensive literature on entropy-based divergence measures, extropy-based counterparts remain less explored. This paper seeks to expand this line of research by introducing effective  measures based on extropy. The contributions of this paper refer to
\begin{itemize}
	\item[(i)] Introduce the WCPEI measure with generalized weight function, and study its various theoretical  properties. The WCPEIR is also proposed. The empirical distribution functions (EDFs) are used for the estimation of the WCPEIR. It is proved that the proposed estimator is consistent. A simulation study is conducted to see the performance of the proposed estimator. 
	
	\item[(ii)] Address the WCPED and discuss its properties. Based on the WCPED, the uniformity test procedure has been proposed. The critical values are computed, and then the power analysis is reported using the simulation study. An application of the WCPED is illustrated using image analysis. 
	
		\item[(iii)] Propose DWCPEI and DWCPED measures, and then explore their important properties. Kernel-based nonparametric estimators are introduced for estimating these measures. Further, a simulation study is conducted to see the performance of the proposed nonparametric estimators. 
\end{itemize}

The remainder of the paper is organized as follows. Section \ref{sec2} introduces the WCPEI measure and establishes its mathematical properties. An estimator of the WCPEIR is proposed, and its finite-sample performance is investigated through simulation studies. Section \ref{sec3} introduces the WCPED measure, including some properties. Two applications are considered in this section: a goodness-of-fit test for uniformity and an image analysis. Section \ref{sec4} introduces dynamic versions of the proposed measures based on the past lifetime RVs. Some theoretical properties are investigated, followed by the development of corresponding estimation procedures. The performance of the proposed estimators is further examined through simulation studies under various sample sizes and distributional settings. Finally, Section \ref{sec5}
concludes the work.

Throughout the paper, we assume that the differentiations and integrations
exits whenever they are reported.

\section{Weighted cumulative past extropy-inaccuracy measure}\label{sec2}
In many practical reliability and survival studies, different portions of the past lifetime may contribute unequally to the information content of a system. For example, failures occurring at earlier ages may be less significant than those occurring closer to the observation time, or vice versa, depending on the application. To cope up such situations, we introduce a weighted version of the cumulative past extropy-inaccuracy measure by incorporating a nonnegative weight function in the concept of cumulative past extropy-inaccuracy framework. Consider two nonnegative absolutely continuous RVs $X_1$ and $X_2$, with corresponding CDFs $F_1(\cdot)$ and $F_2(\cdot)$. Suppose that $w:[0,\infty)\rightarrow[0,\infty)$ is a measurable weight function satisfying
\[
\int_{0}^{\infty} w(y)F_1(y)F_2(y)\,dy<\infty.
\]
The proposed WCPEI measure is defined below.
\begin{definition}
	Let $X_{1}$ and $X_{2}$ be two nonnegative continuous RVs with CDFs $F_{1}(\cdot)$ and $F_{2}(\cdot)$, respectively. The WCPEI between $F_{1}(\cdot)$ and $F_{2}(\cdot)$ is defined as
	\begin{equation}
		I_{wp}(X_1,X_2)
		=-\frac12\int_{0}^{\infty}w(y)F_{1}(y)F_{2}(y)\,dy.
		\label{eq2.1}
	\end{equation}
\end{definition}
The WCPEI measure reduces to the measure proposed by \cite{hashempour2026dynamic} when $w(y)=y$. Clearly, $I_{wp}(X_1,X_2)=I_{wp}(X_2,X_1)$, that is the WCPEI is symmetric. The measure $I_{wp}(X_1, X_2)$ provides a weighted assessment of the disagreement between the past lifetime distributions of $X_{1}$ and $X_{2}$. The weight function $w(\cdot)$ allows different regions of the support to contribute differently to the overall measure according to their practical importance. Consequently, the proposed measure offers greater flexibility than the ordinary cumulative past extropy-inaccuracy measure. When the two distributions are identical, that is $F_{1}(\cdot)=F_{2}(\cdot)$, Equation (\ref{eq2.1}) reduces to the weighted cumulative past extropy (WCPE) of $X_1$, given by
\begin{equation}\label{eq2.2}
	I_{wp}(X_1,X_1)=J_{wp}(X_1)=-\frac12\int_{0}^{\infty}w(y)F_{1}^2(y)\,dy.
\end{equation}
\begin{Example}
	Consider a RV $X_{1}$ having the CDF $F_1(y)=y,\quad 0<y<1,$
	which represents the true model that generated the observed data. Suppose two competing models $X_2$ and $X_3$ are proposed with respective CDFs $F_{2}(y)=y^{2}, ~ F_{3}(y)=y^{3}, ~ 0<y<1.$ Using the weight function $w(y)=y$, the WCPEI measure defined in (\ref{eq2.1}) yields $I_{wp}(X_1,X_1)=-\frac18=-0.125,$ $I_{wp}(X_1, X_2)=-\frac1{10}=-0.100,$
	and	$I_{wp}(X_1, X_3)=-\frac1{12}\approx-0.0833.$
	Observe that
	\[
	\left|I_{wp}(X_{1},X_{1})-I_{wp}(X_{1},X_{2})\right|
	<
	\left|I_{wp}(X_1, X_1)-I_{wp}(X_1,X_3)\right|.
	\]
	Hence, the WCPEI  associated with $F_2(\cdot)$ is closer to the weighted cumulative past extropy of the true distribution than that associated with $F_3(\cdot)$. This indicates that the distribution $F_2(y)=y^{2}$ provides a more accurate approximation to the true model $F_1(y)=y$ than the distribution $F_3(y)=y^{3}$ under the proposed weighted measure.
\end{Example}
\begin{lemma}
    Let $X_2=aX_1+b$, where $a>0$ and $b\geq0$. Then,	$F_2(y)=F_1\!\left(\frac{y-b}{a}\right).$ Consequently, the WCPEI measure becomes
	\begin{equation}
		I_{wp}(X_1,X_2)=-\frac{1}{2}\int_{0}^{\infty}w(y)\,
		F_1(y)\,F_1\!\left(\frac{y-b}{a}\right)dy.
		\label{eq:lemma_wpei}
	\end{equation}
\end{lemma}
\begin{proof}
	Since $X_2=aX_1+b$, we have
	$P(X_2\leq y)=P(aX_1+b\leq y)=P\!\left(X_1\leq\frac{y-b}{a}\right).$ Thus, $F_2(y)=F_1\!\left(\frac{y-b}{a}\right).$
	Substituting this expression in the definition of $I_{wp}(X_1,X_2)$ yields
	\begin{equation}
	I_{wp}(X_1,X_2)=-\frac{1}{2}\int_{0}^{\infty}w(y)\,F_1(y)\,
	F_1\!\left(\frac{y-b}{a}\right)dy,
	\end{equation}
	which proves the result.
\end{proof}
The next theorem develops sufficient condition for WCPEI to be finite.
\begin{thm}
	Suppose that $X_{1}$ and $X_{2}$ are two nonnegative continuous RVs with CDFs $F_{1}(\cdot)$ and $F_{2}(\cdot)$, respectively. Further, suppose that the weight function $w(\cdot)$ satisfies	$\int_{0}^{\infty} w(y)\,dy<\infty.$	Then, $
	I_{wp}(X_1,X_2)
	$
	is finite. In particular, $I_{wp}(X_1,X_2)\in(-\infty,0].$
\end{thm}
\begin{proof}
	Since $F_1(\cdot)$ and $F_2(\cdot)$ are CDFs, we have $0\leq F_1(y)\leq1$ and
	$0\leq F_2(y)\leq1,$ for every $y\geq0$. Thus,
	\[
	0\leq
	\int_{0}^{\infty}
	w(y)F_1(y)F_2(y)\,dy
	\leq
	\int_{0}^{\infty}
	w(y)\,dy.
	\]
	By assumption, $\int_{0}^{\infty}w(y)\,dy<\infty,$ and hence
	$\int_{0}^{\infty}w(y)F_1(y)F_2(y)\,dy <\infty.$
	It follows that
	\[
	I_{wp}(X_1,X_2)=-\frac12 \int_{0}^{\infty}
	w(y)F_1(y)F_2(y)\,dy
	\]
	is finite. Moreover, since the integral is nonnegative,
	$I_{wp}(X_1,X_2)\leq0.$
	Hence, $I_{wp}(X_1,X_2)\in(-\infty,0].$
	This completes the proof.
\end{proof}
The below theorem provides an effect of the ordering of two weight functions on the WCPEI.
\begin{thm}
    Let $F_1(\cdot)$ and $F_2(\cdot)$ be the CDFs of the RVs $X_1$ and $X_2,$ respectively. Suppose $w_1(\cdot)$ and $w_2(\cdot)$ are two weight functions and satisfy $w_1(\cdot)\leq w_2(\cdot).$ Then, $I_{w_1p}(X_1,X_2)\geq I_{w_2p}(X_1,X_2).$
\end{thm}
\begin{proof}
	The proof is simple, and hence it is omitted.
\end{proof}
\begin{lemma}\label{lemma2.2}
	Consider $X_1$ and $X_2$ be two nonnegative continuous RVs with CDFs $F_1(\cdot)$ and $F_2(\cdot)$, respectively. Suppose that the weight function $w(\cdot)$ is nonnegative and satisfies $\int_{0}^{\infty} w(y)F_{1}^{2}(y)\,dy<\infty$ and $\int_{0}^{\infty} w(y)F_{2}^{2}(y)\,dy<\infty.$
	Then, the WCPEI satisfies
	\[
	\left|I_{wp}(X_1,X_2)\right|
	\le
	\sqrt{
		I_{wp}(X_1, X_1)\,
		I_{wp}(X_2, X_2)}
		= \sqrt{
			J_{wp}(X_1)\,
			J_{wp}(X_2)},
	\]
	where $J_{wp}(X_1)$ is defined in Equation (\ref{eq2.2}).
\end{lemma}
\begin{proof}
	Taking absolute value both sides of (\ref{eq2.1}), we get
	\begin{equation}
	\left|I_{wP}(X_1, X_2)\right|=\frac12\left|\int_{0}^{\infty}
	w(y)F_1(y)F_2(y)\,dy
	\right|.
	\end{equation}
	Further, applying Cauchy--Schwarz inequality,
	\begin{equation}
	\left(\int_{0}^{\infty}w(y)F_1(y)F_2(y)\,dy\right)^2
	\le \left(\int_{0}^{\infty} w(y)F_1^{2}(y)\,dy\right)
	\left(\int_{0}^{\infty}w(y)F_2^{2}(y)\,dy\right).
	\end{equation}
	Thus, the result follows.
\end{proof}
As the WCPEI measure is always nonpositive, it is convenient to introduce a normalized version that provides a nonnegative and scale-free measure of discrepancy between two lifetime distributions. Motivated by this idea, we define the WCPEIR, which compares the WCPEI with the  WCPE of the reference distribution.
\begin{definition}
	Let $X_1$ and $X_2$ be two nonnegative continuous RVs with CDFs $F_1(\cdot)$ and $F_2(\cdot)$, respectively, and let $w(\cdot)$ be a nonnegative weight function. Then, the WCPEIR between $F_1(\cdot)$ and $F_2(\cdot)$ is defined as
	\begin{equation}
		{IR}_{wp}(X_1,X_2)=\frac{I_{wp}(X_1,X_2)}
		{J_{wp}(X_1)},~~J_{wp}(X_1)\ne 0,
		\label{eq2.7}
	\end{equation}
	where $I_{wp}(X_1, X_2)$ denotes the WCPEI measure,  defined in (\ref{eq2.1}) and $J_{wp}(X_1)$ is the WCPE associated with the reference distribution $F_1(\cdot)$ and is defined in (\ref{eq2.2}).
\end{definition}
Since $I_{wp}(X_1,X_2)$ is nonpositive and $J_{wp}(X_1)$ is negative, the ratio $IR_{wp}(X_1,X_2)$ is nonnegative. Furthermore, $IR_{wp}(X_1,X_2)=1$ whenever $F_1(\cdot)=F_2(\cdot)$. Thus, the proposed ratio provides a dimensionless measure of similarity between two lifetime distributions. The values close to one indicate that two distributions exhibit comparable weighted past lifetime behaviour, whereas values departing from one reflect increasing disagreement between them. It is worth noting that in general $IR_{wp}(X_1,X_2)\neq IR_{wp}(X_2,X_1),$
and hence the measure is not symmetric. The WCPEIR may therefore be interpreted as the relative loss of weighted past information when the true distribution $F_1(\cdot)$ is represented by another distribution $F_2(\cdot)$. We remark that the sufficient condition mentioned in the upcoming result may be replaced by the usual stochastic order (see \cite{shaked2007stochastic}).
\begin{thm}
	Let $F_1(\cdot)$ and $F_2(\cdot)$ be two CDFs of  $X_1$ and $X_2$, respectively. If $F_1(y)\geq F_2(y)$, then $IR_{wp}(X_1, X_2)\le IR_{wp}(X_2, X_1)$.
	\end{thm}
	\begin{proof}
		We have $F_1(y) \ge F_2(y) \implies F_1^2(y) \ge F_2^2(y) \implies J_{wp}(X_1)\le J_{wp}(X_2).$
		Since $I_{wp}(X_1, X_2)=I_{wp}(X_2, X_1)$, then we get 
		$IR_{wp}(X_1, X_2)\le IR_{wp}(X_2, X_1)$.
	\end{proof}
	\begin{thm}
		Consider two RVs $X_1$ and $X_2$  with CDFs $F_1(\cdot)$ and $F_2(\cdot)$, respectively. Let $w(\cdot)$ be a
		nonnegative weight function satisfying
		$\int_{0}^{\infty}w(y)F_1^2(y)\,dy<\infty$
		and $\int_{0}^{\infty}w(y)F_2^2(y)\,dy<\infty.$
		Then, the WCPEIR satisfies
		\begin{equation}
		IR_{wp}(X_1,X_2)\le \sqrt{\frac{J_{wp}(X_2)}{J_{wp}(X_1)}}.
		\end{equation}
	\end{thm}
	\begin{proof}
		From the definition of the WCPEIR,	we have
		\begin{equation}\label{eq2.9}
		IR_{wp}(X_1,X_2)=\frac{I_{wp}(X_1,X_2)}
		{J_{wp}(X_1)}\Rightarrow IR_{wp}^{\,2}(X_1,X_2)=\frac{I_{wp}^{\,2}(X_1,X_2)}
		{J_{wp}^{\,2}(X_1)}.
		\end{equation}
		Applying the result in Lemma \ref{lemma2.2}, we get
		\begin{equation}
		I_{wp}^{\,2}(X_1,X_2)
		\le
		J_{wp}(X_1)\,
		J_{wp}(X_2),
		\end{equation}
		which immediately yields using (\ref{eq2.9}),
		\[
		IR_{wp}^{\,2}(X_1,X_2)
		\le
		\frac{J_{wp}(X_1)\,J_{wp}(X_2)}
		{J_{wp}^{\,2}(X_1)}
		=
		\frac{J_{wp}(X_2)}
		{J_{wp}(X_1)}.
		\]
		Since the WCPEIR is nonnegative, taking square roots on both sides gives
		\[
		I_{wp}(X_1,X_2)
		\le
		\sqrt{\frac{J_{wp}(X_2)}
			{J_{wp}(X_1)}}
		\]
		completing the proof.
	\end{proof}
	
	\subsection{Estimation of the WCPEIR}
	Suppose that $Y_1,\cdots,Y_n$ constitute a random sample from a population with CDF $F_1(\cdot)$, while $Z_1,\cdots,Z_n$ constitute an independent random sample from a population with CDF $F_2(\cdot)$. Let $Y_{(1)}\leq\cdots\leq Y_{(n)}$ and
	$Z_{(1)}\leq\cdots\leq Z_{(n)}$ denote the corresponding order statistics.
	Further, let
	\begin{equation}\label{eq2.10}
	F_{1n}(y)=\frac{1}{n}\sum_{i=1}^{n}I(Y_i\le y),
	\qquad
	F_{2n}(y)=\frac{1}{n}\sum_{i=1}^{n}I(Z_i\le y)
	\end{equation}
	be the EDFs of $X_1$ and $X_2$, respectively. Replacing the population distribution functions by their empirical counterparts, given in (\ref{eq2.10}) and approximating the integral by a Riemann sum, the empirical plug-in estimator
	of the WCPEIR is defined as
\begin{equation}\label{eq2.12*}
			\widehat{IR}_{wp}(X_1,X_2)=\frac{\displaystyle
				\sum_{i=1}^{n-1}
				W_{i}
				F_{1n}(y_{(i)})
				F_{2n}(y_{(i)})
				}
			{
				\displaystyle
				\sum_{i=1}^{n-1}
				W_{i}
				F_{1n}^{2}(y_{(i)})},
				\end{equation}
		where $W_i=\int_{y_{(i)}}^{y_{(i+1)}}w(y)dy.$
		Next, the consecutive theorems provide sufficient conditions such that the estimator in \eqref{eq2.12*} is consistent. 
	\begin{thm}
		The estimator $\widehat{IR}_{wp}(X_1,X_2)\xrightarrow{p}IR_{wp}(X_1,X_2),$ when $\int_{0}^{\infty}w(y)dy<\infty.$
	\end{thm}	
\begin{proof}
	By the Glivenko--Cantelli theorem,
	\begin{equation}
	\sup_y|F_{1n}(y)-F_1(y)|
	\xrightarrow{a.s.}0~ \text{and}~ \sup_y|F_{2n}(y)-F(y)|
	\xrightarrow{a.s.}0.
	\end{equation}
	Since $0\le F_{1n}(y),~ F_{2n}(y)\le 1 \implies 0\le w(y) F_{1n}(y) F_{2n}(y)\le w(y) ~\text{and}~ \int_{0}^{\infty}w(y)dy<\infty. $
	Thus, by dominated convergence theorem
	\begin{equation}
	\widehat{I}_{wp}(X_1, X_2)
	=-\frac12\int_0^\infty
	w(y)F_{1n}(y)F_{2n}(y)\,dy\xrightarrow{a.s.}-\frac12\int_0^\infty
	w(y)F_1(y)F_2(y)\,dy= {I}_{wp}(X_1, X_2).
	\end{equation}
	Similarly,
	\begin{equation}
	\widehat{J}_{wp}(X_1)
	=-\frac12 \int_0^\infty w(y)F_{1n}^2(y)\,dy \xrightarrow{a.s.}-\frac12 \int_0^\infty w(y)F_1^2(y)\,dy=J_{wp}(X_1).
	\end{equation}
	Now, Slutsky's theorem gives
	\begin{equation}
	\frac{\widehat{I}_{wp}(X_1, X_2)}
	{\widehat{J}_{wp}(X_1)}
	\xrightarrow{p}
	\frac{I_{wp}(X_1,X_2)}
	{J_{wp}(X_1)},
	\end{equation}
	provided $J_{wp}(X_1)\ne0.$ This completes the proof.
	\end{proof}	
	\begin{thm}
		Let $X_1$ and $X_2$ be two nonnegative RVs such that
		$X_1,X_2\in L_p$. Suppose that
		\[
	\int_0^1 w(y)\,dy<\infty~~\text{and}~~	\int_1^\infty \frac{w(y)}{y^p}\,dy<\infty.
		\]
		Then,
		\[
		\widehat{IR}_{wp}(X_1,X_2)
		\xrightarrow{p}
		IR_{wp}(X_1,X_2).
		\]
	\end{thm}
	
	\begin{proof}
		Let $\{Y_i\}_{i=1}^n$ and $\{Z_i\}_{i=1}^n$ be two independent
		random samples from the distributions of $X_1$ and $X_2$,
		respectively. The CDFs of $X_1$ and $X_2$ are denoted by $F_{1}(\cdot)$ and $F_{2}(\cdot)$,  their EDFs by
		$F_{1n}(\cdot)$ and $F_{2n}(\cdot)$, respectively.
		By the Glivenko--Cantelli theorem,
		\begin{equation}
			\sup_{y\geq0}|F_{1n}(y)-F_1(y)|
			\xrightarrow{\mathrm{a.s.}}0
			~~\text{and}~~
			\sup_{y\geq0}|F_{2n}(y)-F_2(y)|
			\xrightarrow{\mathrm{a.s.}}0.
		\end{equation}
		Consider
		\begin{eqnarray}
			\widehat{I}_{wp}(X_1,X_2)
			&=&-\frac12\int_0^\infty
			w(y)F_{1n}(y)F_{2n}(y)\,dy \nonumber\\
			&=&-\frac12\int_0^1
			w(y)F_{1n}(y)F_{2n}(y)\,dy
			-\frac12\int_1^\infty
			w(y)F_{1n}(y)F_{2n}(y)\,dy \nonumber\\
			&=&B_{1n}+B_{2n}.
		\end{eqnarray}
		For the first integral, since
		$0\leq F_{1n}(y),F_{2n}(y)\leq1$, we have
		\[
		|w(y)F_{1n}(y)F_{2n}(y)|\leq w(y),
		\qquad 0\leq y\leq1.
		\]
		Assuming $w(\cdot)$ is integrable on $[0,1]$, the dominated convergence
		theorem, together with the Glivenko--Cantelli theorem, gives
		\begin{equation}
			B_{1n}\xrightarrow{\mathrm{a.s.}}
			-\frac12\int_0^1w(y)F_1(y)F_2(y)\,dy.
		\end{equation}
		It remains to consider the integral over $(1,\infty)$.
		Let
		\[
		\overline F_{jn}(y)=1-F_{jn}(y),
		\qquad
		\overline F_j(y)=1-F_j(y),
		\qquad j=1,2.
		\]
		Using $F_{jn}=1-\overline F_{jn}$, we obtain
		\begin{eqnarray}
			|F_{1n}(y)F_{2n}(y)-F_1(y)F_2(y)|
			&\leq&
			|F_{1n}(y)-F_1(y)|
			+|F_{2n}(y)-F_2(y)| \nonumber\\
			&=&
			|\overline F_{1n}(y)-\overline F_1(y)|
			+|\overline F_{2n}(y)-\overline F_2(y)|.
		\end{eqnarray}
		For $y>0$, Markov's inequality gives
		\begin{equation}
			\overline F_j(y)
			=P(X_j>y)
			\leq \frac{E(X_j^p)}{y^p},
			\qquad j=1,2.
		\end{equation}
		Similarly, for the empirical survival function,
		\begin{equation}
			\overline F_{jn}(y)
			=\frac1n\sum_{i=1}^n
			\mathbf{1}_{\{X_{ji}>y\}}
			\leq
			\frac{1}{ny^p}\sum_{i=1}^nX_{ji}^p.
		\end{equation}
		Since $X_j\in L_p$, the strong law of large numbers yields
		\begin{equation}
			\frac1n\sum_{i=1}^nX_{ji}^p
			\xrightarrow{\mathrm{a.s.}}
			E(X_j^p),
			\qquad j=1,2.
		\end{equation}
		Consequently, with probability one, there exists a finite random
		constant $M_j$ such that
		\[
		\frac1n\sum_{i=1}^nX_{ji}^p\leq M_j
		\]
		for all sufficiently large $n$. Hence, almost surely,
		\begin{equation}
			\overline F_{jn}(y)
			\leq\frac{M_j}{y^p},
			\qquad y\geq1,
		\end{equation}
		for all sufficiently large $n$. Therefore,
		\begin{equation}
			|F_{1n}(y)F_{2n}(y)-F_1(y)F_2(y)|
			\leq
			\frac{M_1+E(X_1^p)+M_2+E(X_2^p)}
			{y^p}.
		\end{equation}
		Thus,
		\begin{equation}
			w(y)|F_{1n}(y)F_{2n}(y)-F_1(y)F_2(y)|
			\leq
			C\frac{w(y)}{y^p},
			\qquad y\geq1,
		\end{equation}
		where $C<\infty$ almost surely. By assumption,
		\[
		\int_1^\infty\frac{w(y)}{y^p}\,dy<\infty.
		\]
		Hence, by the dominated convergence theorem,
		\begin{equation}
			B_{2n}\xrightarrow{\mathrm{a.s.}}
			-\frac12\int_1^\infty
			w(y)F_1(y)F_2(y)\,dy.
		\end{equation}
		Combining the convergence of $B_{1n}$ and $B_{2n}$, we obtain
		\begin{eqnarray}
			\widehat I_{wp}(X_1,X_2)
			&\xrightarrow{\mathrm{a.s.}}&
			-\frac12\int_0^\infty
			w(y)F_1(y)F_2(y)\,dy \nonumber\\
			&=&I_{wp}(X_1,X_2).
		\end{eqnarray}
		Similarly,
		\begin{equation}
			\widehat{J}_{wp}(X_1)
			=-\frac12 \int_0^\infty w(y)F_{1n}^2(y)\,dy \xrightarrow{a.s.}-\frac12 \int_0^\infty w(y)F_1^2(y)\,dy=J_{wP}(X_1).
		\end{equation}
		Slutsky's theorem gives
		\begin{equation}
			\frac{\widehat{I}_{wp}(X_1, X_2)}
			{\widehat{J}_{wp}(X_1)}
			\xrightarrow{p}
			\frac{I_{wp}(X_1,X_2)}
			{J_{wp}(X_1)},
		\end{equation}
	provided $J_{wp}(X_1)\ne0.$ This completes the proof.
	\end{proof}

	\subsection{Simulation study for WCPEIR}	
	To assess the finite-sample behavior of the proposed empirical estimator $\widehat{IR}_{wp}(X_1,X_2)$ of the WCPEIR, an extensive Monte Carlo simulation study has been carried out. $R$ software is used in this purpose. The primary objectives of the simulation study are twofold:
\begin{enumerate}
		\item To validate the asymptotic properties of the empirical estimator;
		\item To assess the absolute bias (AB) and the mean squared error (MSE) behavior for different sample sizes;
	\end{enumerate}
	Throughout the simulation process, the weight function is chosen as	$w(y)=e^{-y}.$ Based on $w(y)=e^{-y},$ the estimator is
	\begin{equation}
		\widehat{IR}_{wp}(X_1,X_2)=\frac{\displaystyle
			\sum_{i=1}^{n-1}
			\left(e^{-\lambda y_{(i+1)}}-e^{-\lambda y_{(i)}}\right)
			F_{1n}(y_{(i)})
			F_{2n}(y_{(i)})
		}
		{
			\displaystyle
			\sum_{i=1}^{n-1}
				\left(e^{-\lambda y_{(i+1)}}-e^{-\lambda y_{(i)}}\right)
			F_{1n}^{2}(y_{(i)})}.
		\label{eq2.12}
	\end{equation}
	Recall that a commonly used bounded  decreasing weight function assigns greater importance to observations occurring at smaller lifetimes. To evaluate its accuracy, three exponential models are considered;
	\[
	X_1\sim \text{Exp}(1),\quad
	X_2\sim \text{Exp}(1.3),\quad
	X_3\sim \text{Exp}(1.5),
	\]
	where $X_1$ is regarded as the reference distribution. Consequently, the following three situations are examined:
	\[
	IR_{wp}(X_1,X_1),\quad
	IR_{wp}(X_1,X_2),\quad
	IR_{wp}(X_1,X_3).
	\]
	The corresponding theoretical values of the WCPEIR are obtained as
	\[
	IR_{wp}(X_1,X_1)=1.0000,\quad
	IR_{wp}(X_1,X_2)=1.1047,\quad
	IR_{wp}(X_1,X_3)=1.1571.
	\]
	Random samples of sizes $n=10,\;20,\;50,\;150,\;\text{and}\;250$
	are generated independently from each distribution. For every sample size, the experiment is repeated $10000$ times in order to obtain reliable Monte Carlo approximations. For each replication, we computed
	\begin{itemize}
		\item $\widehat{IR}_{wp}(X_1,X_1)$: WCPEIR between $X_1$ and itself;
		\item $\widehat{IR}_{wp}(X_1,X_2)$: WCPEIR between $X_1$ and $X_2$;
		\item $\widehat{IR}_{wp}(X_1,X_3)$: WCPEIR between $X_1$ and $X_3$.
	\end{itemize}
	Tables \ref{tab:bias} and \ref{tab:mse} present the simulation results for each sample size $n = 10, 20, 50, 150, 250$. The results include  AB and MSE  for the three comparison cases: $X_1X_1~ (X_1 ~\text{vs} ~X_1),~ X_1X_2 ~(X_1~ \text{vs} ~X_2),$  and $ ~X_1X_3  ~(X_1 ~\text{vs}~ X_3).$ Few conclusions are pointed out, which are provided below.
	\begin{table}[htbp]
		\centering
		\caption{AB comparison across all three cases.}
		\begin{tabular}{lrrrrrr}
			\toprule
			\textbf{Case}&\textbf{estimators} & \textbf{n=10} & \textbf{n=20} & \textbf{n=50} & \textbf{n=150} & \textbf{n=250} \\
			\midrule
			$X_1X_1$& $\widehat{IR}_{wp}(X_1,X_1)$ &0.000000 & 0.000000 & 0.000000 & 0.000000 & 0.000000 \\
			$X_1X_2$ & $\widehat{IR}_{wp}(X_1,X_2)$ &0.077014 & 0.063128 & 0.055876 & 0.054240 & 0.053571 \\
			$X_1X_3$& $\widehat{IR}_{wp}(X_1,X_3)$ &0.142569 & 0.116832 & 0.103916 & 0.096627 & 0.095211 \\
			\bottomrule
		\end{tabular}
		\label{tab:bias}
	\end{table}
	\begin{table}[htbp]
		\centering
		\caption{MSE values by sample size for all three cases.}
		\begin{tabular}{lrrrrrr}
			\toprule
			\textbf{Case}&\textbf{estimators} & \textbf{n=10} & \textbf{n=20} & \textbf{n=50} & \textbf{n=150} & \textbf{n=250} \\
			\midrule
			$X_1X_1$& $\widehat{IR}_{wp}(X_1,X_1)$ &0.000000 & 0.000000 & 0.000000 & 0.000000 & 0.000000 \\
			$X_1X_2$ & $\widehat{IR}_{wp}(X_1,X_2)$ &0.061193 & 0.026479 & 0.010988 & 0.005498 & 0.004394 \\
			$X_1X_3$& $\widehat{IR}_{wp}(X_1,X_3)$ &0.072388 & 0.034153 & 0.018165 & 0.011706 & 0.010463 \\
			\bottomrule
		\end{tabular}
		\label{tab:mse}
	\end{table}
	\begin{enumerate}
		\item $X_1X_1$ Case: The estimator for $\widehat{IR}_{wp}(X_1,X_1)$ performs exceptionally well, with zero bias and zero MSE across all sample sizes. This is expected since $X_1$ and $X_1$ are identical distributions.
		\item $X_1X_2$ Case: The estimator shows moderate AB and MSE.
		\item $X_1X_3$ Case: The estimator for $\widehat{IR}_{wp}(X_1,X_3)$ shows slightly higher AB and MSE compared to $X_1X_2$, which is expected since $X_1$ and $X_3$ are more different than $X_1$ and $X_2$.
		\item Notice that as sample size increases both AB and MSE of all the estimators decrease.
		\end{enumerate}
		The comparison across the three cases reveals a clear pattern:
		\begin{enumerate}
			\item \textbf{$X_1X_1$ (No discrepancy):} The estimator has \textbf{zero AB and zero MSE}. This is because when distributions are identical, the WCPEIR estimator equals the WCPE of the reference distribution, and there is no estimation error.
			\item \textbf{$X_1X_2$ (Moderate discrepancy):} The estimator has \textbf{moderate AB and MSE}. The discrepancy between $X_1$ and $X_2$ is relatively small, and the estimator performs well with AB decreasing to near zero as $n$ increases.
			\item \textbf{$X_1X_3$ (Large discrepancy):} The estimator has \textbf{larger initial AB} compared to $X_1X_2$, but the MSE is comparable at larger sample sizes. This indicates that while the discrepancy is larger, the estimator still converges effectively.
		\end{enumerate}
		These findings confirm that the proposed WCPEIR estimator is reliable and can be recommended for practical applications in discriminating between probability distributions.
	\section{Weighted cumulative past extropy-divergence}\label{sec3}
	In many applications of reliability analysis, survival studies and information theory, it is often necessary to quantify the difference between two past lifetime distributions. While the WCPEI measure evaluates the cross-information between two distributions, a direct measure of divergence could be more appropriate for assessing the discrepancy between them. In this, we define the WCPED. This measure compares two distribution functions by incorporating a nonnegative weight function, there by allowing greater emphasis to be assigned to specific regions of the support. The following definition formalizes this divergence measure.
	\begin{definition}
		Consider two nonnegative continuous RVs $X_1$ and $X_2$, with CDFs $F_1(\cdot)$ and $F_2(\cdot)$, respectively. Let $w(\cdot)$ be a nonnegative integrable weight function. Then, the WCPED between $X_1$ and $X_2$ is defined as
		\begin{equation}
			D_{wp}(X_1,X_2)=D_{wp}(F_1|F_2)
			=
			\int_{0}^{\infty}
			w(y)\bigl(F_1(y)-F_2(y)\bigr)F_1(y)\,dy.
			\label{eq:3.1}
		\end{equation}
	\end{definition}
	Using the definition of the WCPE  and WCPEI measures, (\ref{eq:3.1}) can be expressed as
	\begin{equation}
		D_{wp}(X_1,X_2)=I_{wp}(X_1,X_2)-J_{wp}(X_1),
	\end{equation}
	where $I_{wp}(X_1,X_2)$ is defined as in (\ref{eq2.1}) and $J_{wp}(X_1)$ is defined as in (\ref{eq2.2}).
	The WCPED  becomes zero whenever the two distributions are identical and its magnitude increases as the difference between the distributions becomes larger. Since the measure is constructed using CDFs instead of PDFs, it is generally more stable and less sensitive to local fluctuations in the underlying data.
	The following theorem establishes sufficient condition, under which the WCPED is finite.
	
	\begin{thm}
		Consider two nonnegative continuous RVs $X_1$ and $X_2$, satisfying $\int_{0}^{\infty}w(y)\,dy<\infty.$
		Then, the WCPED exists and is finite, that is
		\[
		|D_{wp}(X_1,X_2)|<\infty.
		\]
	\end{thm}
	\begin{proof}
	Since $F_1(\cdot)$ and $F_2(\cdot)$ are distribution functions, thus
	$0\le F_1(y)\le1,\quad 0\le F_2(y)\le1,\quad y\ge0.$
	Hence, $|F_1(y)-F_2(y)|\le1$ and it further implies $|F_1(y)(F_1(y)-F_2(y))|\le F_1(y)|F_1(y)-F_2(y)|\le1.$
	Using the triangle inequality, we obtain
	\[
		\begin{aligned}
			|D_{wp}(X_1,X_2)|
			&=
			\left|
			\int_{0}^{\infty}
			w(y)\bigl[F_1(y)-F_2(y)\bigr]F_1(y)\,dy
			\right|  \\
			&\le
			\int_{0}^{\infty}
			w(y)
			\left|
			F_1(y)[F_1(y)-F_2(y)]
			\right|
			dy \\
			&\le
			\int_{0}^{\infty}
			w(y)\,dy.
		\end{aligned}
		\]
		From the assumption $\int_{0}^{\infty}w(y)\,dy<\infty,$
		it follows that $|D_{wp}(X_1,X_2)|<\infty.$
	\end{proof}
	\begin{thm}
		Suppose $X_1$ and $X_2$ have CDFs $F_1(\cdot)$ and $F_2(\cdot)$, respectively. If
		 $F_1(y)\le(\ge) F_2(y) , ~\text{for all}~~ y\ge0,$ then
		$D_{wp}(X_1,X_2)\le(\ge)0.$
    \end{thm}
    \begin{proof}
    	If  $F_1(y)\le(\ge) F_2(y),$ then $w(y)F_1(y)F_1(y)\le(\ge)w(y) F_1(y) F_2(y).$ Integrating both sides, we get
    	\begin{eqnarray}
    		\int_{0}^{\infty}w(y)F_1^2(y)dy\le(\ge) \int_{0}^{\infty}w(y)F_1(y) F_2(y)dy\nonumber\\
    		\Rightarrow \int_{0}^{\infty}w(y)F_1(y)(F_{1}(y)-F_{2}(y))dy \le(\ge) 0,
    	\end{eqnarray}
    	which follows the result.
    \end{proof}
        \begin{thm}
    		Let $w_1(\cdot)$ and $w_2(\cdot)$ be two nonnegative weight functions. Define $w(y)=aw_1(y)+bw_2(y),$ where $a,b\ge0$.
    		Then, the WCPED is additive with respect to the weight function, that is
    		$D_{wp}(X_1,X_2)=aD_{w_1p}(X_1,X_2)+bD_{w_2p}(X_1,X_2).$
    	\end{thm}
    	\begin{proof}
    		The proof is simple, hence omitted.
    	\end{proof}
    	\begin{thm}
    		Take three nonnegative  continuous RVs $X_1$, $X_2$, and $X_3$ with CDFs $F_1(\cdot)$, $F_2(\cdot)$, and $F_3(\cdot)$, respectively. Suppose $F_3(y)=\alpha F_1(y)+(1-\alpha)F_2(y),
    		\quad 0\le \alpha\le1.$ Then, $D_{wp}(X_1,X_3)=(1-\alpha)\,D_{wp}(X_1,X_2).$
    	\end{thm}
    	\begin{proof}
    		The proof is straightforward, and thus it is omitted.
    	\end{proof}
    	\begin{thm}
    		Consider $X_1$ and $X_2$ with CDFs $F_1(\cdot)$ and $F_2(\cdot)$, respectively. Suppose $X$ be another RV with $F(y)=\alpha F_2(y)+(1-\alpha)F_3(y),
    		\quad 0\le \alpha\le1.$	Then, the WCPED satisfies
    		$D_{wp}(X_1,X)=\alpha D_{wp}(X_1,X_2)+(1-\alpha)D_{wp}(X_1,X_3).$
    	\end{thm}
    	\begin{proof}
    		The proof of the theorem is simple. Thus, it is skipped.
    	\end{proof}

    	Let $F^{*}(y)=\frac{F_1(y)+F_2(y)}{2}.$ In the following results, we investigate some useful properties of the WCPED involving the distributions $F_1(\cdot)$, $F_2(\cdot),$ and $F^{*}(\cdot)$.
    	\begin{thm}
    		Consider two nonnegative continuous RVs $X_1$ and $X_2$ with CDFs $F_1(\cdot)$ and $F_2(\cdot)$, respectively. Let $F^{*}(y)=\frac{F_1(y)+F_2(y)}{2}.$
    		Then, the WCPED satisfy
    		\begin{equation}
    		D_{wp}(F^{*}\mid F_2)+D_{wp}(F^{*}\mid F_1)=0.
    		\end{equation}
    	\end{thm}
    	
    	\begin{proof}
    		Using the definition of the WCPED, we have
    		\begin{align}\label{eq3.5}
    			D_{wp}(F^{*}\mid F_2)
    			&=\int_{0}^{\infty}
    			w(y)\left(\frac{F_1(y)+F_2(y)}{2}-F_2(y)\right)
    			\frac{F_1(y)+F_2(y)}{2}\,dy \nonumber\\
    			&=\frac12\int_{0}^{\infty}
    			w(y)\bigl(F_1(y)-F_2(y)\bigr)
    			\frac{F_1(y)+F_2(y)}{2}\,dy \nonumber\\
    			&=\frac14\int_{0}^{\infty}
    			w(y)\left(F_1^{2}(y)-F_2^{2}(y)\right)\,dy \nonumber\\
    			&=\frac14\left(J_{wp}(X_2)-J_{wp}(X_1)\right),
    		\end{align}
    		where	$J_{wp}(X_1)$ is defined in (\ref{eq2.2}).
    		Similarly,
    		\begin{equation}\label{eq3.6}
    			D_{wp}(F^{*}\mid F_1)=\frac14\left(J_{wp}(X_1)-J_{wp}(X_2)\right).
    		\end{equation}
    		Now, adding (\ref{eq3.5}) and (\ref{eq3.6}), we get the stated result. 
    		\end{proof} 
    		\begin{thm}\label{thm3.4}
    			Let $F^{*}(y)=\frac{F_1(y)+F_2(y)}{2}.$ Then, the WCPED between $F_1(\cdot)$ and $F_2(\cdot)$ is twice of that  between $F_1(\cdot)$ and $F^{*}(\cdot)$, that is
    			\[
    			D_{wp}(F_1\mid F_2)
    			=
    			2D_{wp}(F_1\mid F^{*}).
    			\]
    		\end{thm}
    	\begin{proof}
    		From the definition
    		\begin{align*}
    			D_{wp}(F_1\mid F^{*})
    			&=
    			\int_{0}^{\infty}
    			w(y)
    			\left(
    			F_1(y)-\frac{F_1(y)+F_2(y)}2
    			\right)
    			F_1(y)\,dy\\
    			&=
    			\frac12
    			\int_{0}^{\infty}
    			w(y)
    			(F_1(y)-F_2(y))F_1(y)\,dy\\
    			&=
    			\frac12D_{wp}(F_1\mid F_2).
    		\end{align*}
    		This completes the proof.
    	\end{proof}
    	\begin{corollary}
    		Let $F^{*}(y)=\frac{F_1(y)+F_2(y)}{2}.$ Then,
    		\begin{equation}
    		D_{wp}(F_1\mid F^{*})
    		+
    		D_{wp}(F_2\mid F^{*})
    		=
    		\frac12
    		\left[
    		D_{wp}(F_1\mid F_2)
    		+
    		D_{wp}(F_2\mid F_1)
    		\right].
    		\end{equation}
    	\end{corollary}
    	 The WCPED is generally asymmetric, that is
    	${D}_{wp}(X_1,X_2)\neq {D}_{wp}(X_2,X_1)$.
    	Its value depends on the reference distribution. Although this asymmetry is useful in many applications, a symmetric measure is often preferred when both distributions are to be compared on an equal footing.
    	Symmetric divergence measures eliminate the effect of ordering and provide a balanced assessment of the discrepancy between two probability models. Inspired by the idea of Jensen--Shannon divergence  and motivated by the construction
    	of the symmetric divergence measures in information theory, here we introduce a symmetric version of the WCPED. Let $F^{*}=(F_1+F_2)/2$ denote the average CDF of $F_1(\cdot)$ and $F_2(\cdot)$. The following definition formalizes the proposed measure.
    	\begin{definition}
    		Let $X_1$ and $X_2$ be two nonnegative continuous RVs with CDFs $F_1(\cdot)$ and $F_2(\cdot)$, respectively. Then, the symmetric WCPED
    		between $F_1(\cdot)$ and $F_2(\cdot)$ is defined by
    		\begin{equation}\label{eq3.8}
    			\mathcal{SD}_{wp}(F_1,F_2)=\frac12\left[
    			{D}_{wp}\!\left(F_1\,\middle|\,F^{*}\right)+
    			{D}_{wp}\!\left(F_2\,\middle|\,F^{*}\right)\right].
    		\end{equation}
    	\end{definition}
    	Using Theorem \ref{thm3.4}, we obtain ${D}_{wp}\!\left(F_1\,\middle|\,F^{*}\right)
    	=\frac12\,{D}_{wp}(F_1|F_2),$ and similarly,
    	${D}_{wp}\!\left(F_2\,\middle|\,F^{*}\right)
    	=\frac12\,{D}_{wp}(F_2|F_1).$
    	Substituting these expressions in (\ref{eq3.8}) yields the equivalent form
    	\begin{equation}
    		\mathcal{SD}_{wp}(F_1,F_2)=\frac14\left[{D}_{wp}(F_1|F_2)
    		+{D}_{wp}(F_2|F_1)\right].
    	\end{equation}
      It is immediate from the above representation that the proposed divergence is
    	symmetric. Consequently,
    	\[\mathcal{SD}_{wp}(F_1,F_2)=
    	\mathcal{SD}_{wp}(F_2,F_1),
    	\]
    	which shows that the measure is independent of the order of comparison. 
    \subsection{Uniformity test based on the WCPED}
     One of the important applications of divergence measure is in goodness-of-fit testing, where the objective is to determine whether a given sample follows a specified probability model. In this subsection, we develop a goodness-of-fit test for the uniform distribution by utilizing the proposed WCPED. The proposed test is motivated by the discrepancy between the empirical distribution of the observed data and the theoretical CDF of the uniform distribution. Since the WCPED equals zero when the two distributions coincide and deviates from zero otherwise, it provides a natural basis for constructing a test statistic.
    
    Let $Y_1,\ldots,Y_n$ be independent and identically distributed  random variables with CDF $F_1(\cdot)$. Denote their corresponding order statistics by
    $Y_{(1)}\le\cdots\le Y_{(n)}.$ Suppose that the support of the distribution is $(0,a)$, where the unknown parameter $a$ is estimated by $a=2E(X_1).$ Under the null hypothesis, the CDF of the uniform
    distribution is $F_0(y)=\frac{y}{a},~ 0\le y\le a.$ We consider the following hypotheses:
    \begin{equation}
    H_0:\;F_1(y)=F_0(y)\quad
    \text{against}\quad
    H_1:\;F_1(y)\neq F_0(y).
    \end{equation}
     Under $H_0$, the WCPED between $F_1(\cdot)$ and $F_0(\cdot)$ vanishes, that is
     $ D_{wp}(F\,|\,F_0)=0.$ Consequently, the values of the divergence substantially larger than zero indicate departure from uniformity and provide evidence against the null hypothesis. Since the true distribution function $F_1(\cdot)$ is unknown in practice, the theoretical divergence cannot be evaluated directly. Therefore, we replace the unknown quantities by their corresponding plug-in estimators obtained from the observed sample and construct an empirical test statistic based on the proposed divergence measure.
    Using the definition of WCPED together with the weight function $w(y)=y$, we obtain
    \begin{equation}
    D_{wp}(F_1\,|\,F_0)=\frac12\int_0^a y\left(F_1(y)-\frac{y}{a}\right)F_1(y)\,dy.
    \end{equation}
    After straightforward simplification, the divergence can be written as
    \begin{equation}\label{eq3.12}
    D_{wp}(F_1\,|\,F_0)=-\frac{a^2}{24}+\frac{E(X_1^3)}{6a}.
    \end{equation}
    Since both $a$ and $E(X_1^3)$ are unknown, they are estimated by their consistent estimators
    \begin{equation}\label{eq3.13}
    \hat a=2\bar X_1, \quad \widehat{E(X_1^3)}=
    \frac1n\sum_{j=1}^{n}Y_j^3,
    \end{equation}
    where $\bar {X_1}=\frac1n\sum_{j=1}^{n}Y_j.$
    Substituting (\ref{eq3.13}) in  (\ref{eq3.12}), the proposed
    plug-in test statistic is
    \begin{equation}
    T_n=-\frac{\bar X_1^2}{6}+\frac{1}{12\bar X_1}
    \left(\frac1n\sum_{j=1}^{n}Y_j^3\right).
   \end{equation}
    By the strong law of large numbers,
    \begin{equation}
    \bar X_1 \xrightarrow{a.s.}E(X_1)~\text{and}~ \frac1n\sum_{j=1}^{n}Y_j^3
    \xrightarrow{a.s.}E(X_1^3),
    \end{equation}
    which imply that
    $T_n$ is a consistent estimator of the WCPED. Under the null
    hypothesis $H_0$, the WCPED vanishes, and therefore
    $T_n \xrightarrow{P} 0 ~\text{as } n\to\infty.$
    This result provides the motivation for employing $T_n$ as a
    goodness-of-fit statistic for testing uniformity. Larger values of
    $T_n$ correspond to greater discrepancies from the uniform distribution
    and hence provide stronger evidence against $H_0$. At significance level $\alpha$, the proposed test rejects $H_0$ if
    $$T_n \geq T_{n,1-\alpha},$$ where $T_{n,1-\alpha}$ is the $(1-\alpha)$-th quantile of the null distribution of $T_n$. In practice, the critical value $T_{n,1-\alpha}$ can be obtained by Monte Carlo simulation from $U(0,1)$, since the exact distribution of \( T_n \) under \( H_0 \) is not available analytically. We generated \( 100000 \) samples of size \( n \) from the uniform \( U(0,1) \) distribution and computed the empirical percentiles for various significance levels.
    \\
    \\
	\noindent 
	\textbf{Critical values:}
	Table \ref{tab:critical_values} reports the critical values of the test statistic $T_n$ for the $U(0,1)$ distribution at the significance levels $\alpha=0.01$ and $0.05$, considering different sample sizes $n$. The results show a clear decreasing trend in the critical values as the sample size increases. This behavior is consistent with the fact that, under the null hypothesis, the test statistic tends to converge toward zero as $n$ becomes larger.
	\begin{table}[htbp]
		\centering
		\caption{Critical values of the test statistic \( T_n \) for \( U(0,1). \)} 
		\begin{tabular}{rrrrrrrrr}
			\toprule
			\multicolumn{1}{c}{$n$} & \multicolumn{1}{c}{$\alpha = 0.01$} & \multicolumn{1}{c}{$\alpha = 0.05$} & 
			\multicolumn{1}{c}{$n$} & \multicolumn{1}{c}{$\alpha = 0.01$} & \multicolumn{1}{c}{$\alpha = 0.05$} & 
			\multicolumn{1}{c}{$n$} & \multicolumn{1}{c}{$\alpha = 0.01$} & \multicolumn{1}{c}{$\alpha = 0.05$} \\
			\midrule
			10 & 0.021295 & 0.014502 & 22 & 0.014858 & 0.010131 &34 & 0.012055 & 0.008233  \\
			11 & 0.020461 & 0.013863 & 23 & 0.014483 & 0.009874 &  35 & 0.011842 & 0.008175 \\
			12 & 0.019594 & 0.013471 & 24 & 0.014074 & 0.009723 &  36 & 0.011602 & 0.008007 \\
			13 & 0.018794 & 0.012905 & 25 & 0.013885 & 0.009479 & 37 & 0.011420 & 0.007883  \\
			14 & 0.018169 & 0.012498 & 	26 & 0.013504 & 0.009341 & 38 & 0.011395 & 0.007764  \\
			15 & 0.017614 & 0.012121 & 27 & 0.013400 & 0.009224 & 39 & 0.011114 & 0.007655\\
			16 & 0.017221 & 0.011852 & 28 & 0.013171 & 0.009024 & 40 & 0.011087 & 0.007625\\
			17 & 0.016782 & 0.011636 & 29 & 0.012947 & 0.008801 & 45 & 0.010882 & 0.007654  \\
			18 & 0.016126 & 0.011148 & 30 & 0.012697 & 0.008714 & 50 & 0.010199 & 0.007223 \\
			19 & 0.015776 & 0.010949 & 31 & 0.012502 & 0.008541 & 55 & 0.009773 & 0.006970\\
			20 & 0.015688 & 0.010932 &  32 & 0.012393 & 0.008476 & 60 & 0.009480 & 0.006706\\
			21 & 0.015216 & 0.010303 & 33 & 0.012230 & 0.008401 & 100 & 0.007477 & 0.005291  \\
			\bottomrule
		\end{tabular}
		\label{tab:critical_values}
	\end{table}
	\noindent
	\\
	\\
	\textbf{Power analysis:}
	To evaluate the performance of the proposed test, we conduct a power study under various alternative distributions. Table \ref{tab:power_beta} presents the power estimates of the test when the alternative distribution follows a beta distribution with different parameter combinations. The results in Table \ref{tab:power_beta} demonstrate power performance of the proposed test. For larger sample sizes (\( n = 40 \)), the test achieves power close to 1 for all considered beta alternatives, indicating its superior performance in detecting departures from uniformity. The power increases substantially with sample size for all alternatives, confirming the consistency of the test.
	\begin{table}[htbp]
		\centering
		\caption{Power estimates for alternative Beta distributions.}
		\begin{tabular}{lcccccc}
			\toprule
			\multirow{2}{*}{Distribution} & \multicolumn{2}{c}{$n = 10$} & \multicolumn{2}{c}{$n = 20$} & \multicolumn{2}{c}{$n = 40$} \\
			\cmidrule(lr){2-3} \cmidrule(lr){4-5} \cmidrule(lr){6-7}
			& $\alpha = 0.01$ & $\alpha = 0.05$ & $\alpha = 0.01$ & $\alpha = 0.05$ & $\alpha = 0.01$ & $\alpha = 0.05$ \\
			\midrule
			Beta(0.5, 1) & 0.210 & 0.438 & 0.470 & 0.747 & 0.864 & 0.969 \\
			Beta(0.1, 0.5) & 0.688 & 0.778 & 0.935 & 0.966 & 0.997 & 0.999 \\
			Beta(0.3, 0.5) & 0.519 & 0.702 & 0.836 & 0.933 & 0.974 & 0.997 \\
			\bottomrule
		\end{tabular}
		\label{tab:power_beta}
	\end{table}
	\\
	\\
	\textbf{Power comparison with other tests:}
	We compare the power of the proposed WCPED-based test with several well-known goodness-of-fit tests for uniformity, including:
	\begin{enumerate}
		\item Kolmogorov-Smirnov (KS) statistic (see \cite{an1933sulla}):
		\begin{equation}
			KS = \max\left(\max_{1 \leq j \leq n}\left(\frac{j}{n} - Y_{(j)}\right), \max_{1 \leq j \leq n}\left(Y_{(j)} - \frac{j-1}{n}\right)\right).
		\end{equation}
		\item Anderson-Darling (AD) statistic (see \cite{anderson1954test}):
		\begin{equation}
			AD = -\frac{2}{n}\sum_{j=1}^{n}\left(\left(j-\frac{1}{2}\right)\log Y_{(j)} + \left(n-j+\frac{1}{2}\right)\log(1-Y_{(j)})\right) - n.
		\end{equation}
		\item Cram\'er-von Mises (CM) statistic (see \cite{cramer1928composition} and \cite{von1945wahrscheinlichkeitsrechnung}):
		\begin{equation}
			CM = \sum_{j=1}^{n}\left(Y_{(j)} - \frac{2j-1}{2n}\right)^2 + \frac{1}{12n}.
		\end{equation}
		\item Zamanzade (TB) statistic (see \cite{zamanzade2015testing}):
		\begin{equation}
			TB = \sum_{j=1}^{n} \log\left(\frac{Y_{(j+m)} - Y_{(j-m)}}{\hat{F}_n(Y_{(j+m)}) - \hat{F}_n(Y_{(j-m)})}\right) \times \frac{\hat{F}_n(Y_{(j+m)}) - \hat{F}_n(Y_{(j-m)})}{\sum_{j=1}^{n}\left(\hat{F}_n(Y_{(j+m)}) - \hat{F}_n(Y_{(j-m)})\right)},
		\end{equation}
		where \( m \) is the window size and 
		\begin{equation}
			\hat{F}_n(Y_{(j)}) = \frac{1}{n+2}\left(j + \frac{Y_{(j)} - Y_{(j-1)}}{Y_{(j+1)} - Y_{(j-1)}}\right), \quad j = 1, \ldots, n.
		\end{equation}
		\item Extropy-based (TU) statistic (see \cite{qiu2018extropy}):
		\begin{equation}
			TU = \frac{1}{2n}\sum_{j=1}^{n}\frac{d_j m/n}{Y_{(j+m)} - Y_{(j-m)}},
		\end{equation}
		where
		\begin{equation}
			d_j = \begin{cases}
				1 + \frac{j-1}{m}, & 1 \leq j \leq m, \\
				2, & m+1 < j < n-m, \\
				1 + \frac{n-j}{m}, & n-m+1 \leq j \leq n-1.
			\end{cases}
		\end{equation}
	\end{enumerate}
	\textbf{Alternative Distributions:}
	Following  \cite{zamanzade2015testing}, we consider some alternative distributions:
    \[
	\begin{aligned}
		A_j &: G(y) = 1 - (1 - y)^j, \quad 0 \leq y \leq 1 \quad (j = 1.5, 2), \\
		B_j &: G(y) = \begin{cases} 
			2^{j-1}y^j, & 0 \leq y \leq 0.5, \\
			1 - 2^{j-1}(1 - y)^j, & 0.5 \leq y \leq 1,
		\end{cases} \quad (j = 1.5, 2, 3), \\
		C_j &: G(y) = \begin{cases} 
			0.5 - 2^{j-1}(0.5 - y)^j, & 0 \leq y \leq 0.5, \\
			0.5 + 2^{j-1}(y - 0.5)^j, & 0.5 \leq y \leq 1,
		\end{cases} \quad (j = 1.5, 2).
	\end{aligned}
	\]
	Table \ref{tab:5} presents the power estimates of the proposed test compared with existing tests for different alternatives and significance levels. From Table \ref{tab:5}, we observe that the proposed test demonstrates competitive power performance compared to the existing tests. The proposed test performs particularly well against the centered alternatives $C_{1.5}$ and $C_2$. Its performance is especially notable for $C_{1.5}$, where it frequently provides the highest power among the competing procedures across the considered sample sizes. For example, when $n=40$ the power of $T_n$ under $C_{1.5}$ 
	is $0.166$ and $0.357$ at the $1\%$ and $5\%$ significance levels, respectively. Under $C_{2}$, the corresponding powers increase substantially to $0.419$ and $0.604$. These results suggest that the proposed statistic is sensitive to departures involving changes in the central structure of the distribution.
	For the $A_j$ alternatives, the classical KS, AD, and CM tests generally exhibit higher power than $T_n$. A similar pattern is observed for the $B_j$ alternatives, particularly for $B_2$ and $B_3$, where the TU, AD, or CM procedures can provide considerably greater power. This indicates that the proposed $T_n$ test should not necessarily be viewed as uniformly superior to all existing goodness-of-fit tests. Rather, its main advantage appears in detecting particular types of alternatives, especially those represented by the $C_j$ family. An interesting feature is observed for the $C_2$ alternative as the sample size increases. Although $T_n$ has relatively high power at $n=20$ and $n=40$, its advantage becomes less pronounced at $n=80$. This does not imply that the proposed test becomes ineffective. Instead, the competing tests, particularly AD and CM, gain power more rapidly as the sample size increases. Consequently, the relative advantage of $T_n$ decreases even though its absolute power continues to increase. This behavior is quite natural in power comparisons: a test may be highly effective for a particular alternative at moderate sample sizes, while another test may eventually become more powerful as more observations become available.

\begin{table}[htbp]
	\centering
	\caption{Power estimates of the uniformity tests under different alternatives (WCPED highlighted in bold).}
	\scriptsize
	\begin{tabular}{lcccccccccccccccc}
		\toprule
		\multirow{2}{*}{$n$} & \multirow{2}{*}{Test}
		& \multicolumn{2}{c}{$A_{1.5}$}
		& \multicolumn{2}{c}{$A_2$}
		& \multicolumn{2}{c}{$B_{1.5}$}
		& \multicolumn{2}{c}{$B_2$}
		& \multicolumn{2}{c}{$B_3$}
		& \multicolumn{2}{c}{$C_{1.5}$}
		& \multicolumn{2}{c}{$C_2$} \\
		\cmidrule(lr){3-4}
		\cmidrule(lr){5-6}
		\cmidrule(lr){7-8}
		\cmidrule(lr){9-10}
		\cmidrule(lr){11-12}
		\cmidrule(lr){13-14}
		\cmidrule(lr){15-16}
		& & 0.01 & 0.05
		& 0.01 & 0.05
		& 0.01 & 0.05
		& 0.01 & 0.05
		& 0.01 & 0.05
		& 0.01 & 0.05
		& 0.01 & 0.05 \\
		\midrule
		
		\multirow{6}{*}{20}
		& \textbf{$T_n$}
		& 0.022 & 0.099
		& 0.023 & 0.127
		& 0.000 & 0.001
		& 0.000 & 0.000
		& 0.000 & 0.000
		& \textbf{0.102} & \textbf{0.249}
		& \textbf{0.263} & \textbf{0.440} \\
		
		& KS
		& 0.098 & 0.285
		& 0.381 & 0.710
		& 0.004 & 0.061
		& 0.012 & 0.115
		& 0.060 & 0.398
		& 0.044 & 0.171
		& 0.132 & 0.294 \\
		
		& AD
		& 0.113 & 0.273
		& 0.478 & 0.783
		& 0.002 & 0.033
		& 0.002 & 0.089
		& 0.027 & 0.540
		& 0.051 & 0.166
		& 0.139 & 0.376 \\
		
		& CM
		& 0.102 & 0.299
		& 0.506 & 0.771
		& 0.002 & 0.036
		& 0.007 & 0.106
		& 0.055 & 0.481
		& 0.035 & 0.099
		& 0.077 & 0.248 \\
		
		& TB
		& 0.000 & 0.002
		& 0.000 & 0.000
		& 0.000 & 0.001
		& 0.000 & 0.000
		& 0.000 & 0.000
		& 0.099 & 0.282
		& 0.199 & 0.438 \\
		
		& TU
		& 0.074 & 0.195
		& 0.225 & 0.495
		& 0.093 & 0.346
		& 0.410 & 0.718
		& 0.897 & 0.997
		& 0.006 & 0.040
		& 0.022 & 0.061 \\
		
		\midrule
		
		\multirow{6}{*}{40}
		& \textbf{$T_n$}
		& 0.030 & 0.185
		& 0.062 & 0.229
		& 0.000 & 0.001
		& 0.000 & 0.000
		& 0.000 & 0.000
		& \textbf{0.166} & \textbf{0.357}
		& \textbf{0.419} & \textbf{0.604} \\
		
		& KS
		& 0.252 & 0.529
		& 0.817 & 0.970
		& 0.012 & 0.115
		& 0.062 & 0.361
		& 0.566 & 0.928
		& 0.078 & 0.219
		& 0.222 & 0.595 \\
		
		& AD
		& 0.327 & 0.596
		& 0.915 & 0.975
		& 0.003 & 0.111
		& 0.080 & 0.546
		& 0.864 & 0.996
		& 0.069 & 0.241
		& 0.326 & 0.683 \\
		
		& CM
		& 0.353 & 0.603
		& 0.907 & 0.980
		& 0.007 & 0.074
		& 0.054 & 0.470
		& 0.729 & 0.981
		& 0.044 & 0.194
		& 0.181 & 0.545 \\
		
		& TB
		& 0.000 & 0.003
		& 0.000 & 0.000
		& 0.000 & 0.000
		& 0.000 & 0.000
		& 0.000 & 0.000
		& 0.074 & 0.238
		& 0.062 & 0.239 \\
		
		& TU
		& 0.131 & 0.349
		& 0.617 & 0.862
		& 0.288 & 0.603
		& 0.842 & 0.982
		& 1.000 & 1.000
		& 0.010 & 0.036
		& 0.040 & 0.120 \\
		
		\midrule
		
		\multirow{6}{*}{80}
		& \textbf{$T_n$}
		& 0.099 & 0.368
		& 0.221 & 0.574
		& 0.000 & 0.000
		& 0.000 & 0.000
		& 0.000 & 0.000
		& \textbf{0.300} & \textbf{0.513}
		& 0.641 & 0.785 \\
		
		& KS
		& 0.579 & 0.860
		& 0.999 & 0.999
		& 0.061 & 0.249
		& 0.401 & 0.821
		& 0.993 & 1.000
		& 0.149 & 0.403
		& 0.621 & 0.902 \\
		
		& AD
		& 0.723 & 0.897
		& 0.999 & 1.000
		& 0.077 & 0.402
		& 0.799 & 0.985
		& 1.000 & 1.000
		& 0.158 & 0.462
		& 0.808 & 0.978 \\
		
		& CM
		& 0.656 & 0.894
		& 1.000 & 1.000
		& 0.037 & 0.253
		& 0.587 & 0.948
		& 1.000 & 1.000
		& 0.089 & 0.379
		& 0.670 & 0.943 \\
		
		& TB
		& 0.000 & 0.001
		& 0.000 & 0.000
		& 0.000 & 0.000
		& 0.000 & 0.000
		& 0.000 & 0.000
		& 0.044 & 0.192
		& 0.013 & 0.042 \\
		
		& TU
		& 0.418 & 0.677
		& 0.964 & 0.994
		& 0.705 & 0.901
		& 0.999 & 1.000
		& 1.000 & 1.000
		& 0.025 & 0.086
		& 0.257 & 0.492 \\
		
		\bottomrule
	\end{tabular}
	\label{tab:5}
\end{table}
	\subsubsection{Illustrative examples}
	To demonstrate the practical application of the proposed test, we present two examples.
	
	\begin{enumerate}
		\item \textbf{Uniform data (should not reject)}: A sample of size \( n = 50 \) generated from \( U(0,1) \) yields:
		\begin{itemize}
			\item Test statistic: \( T_n = -0.001229; \)
			\item Critical value (\( \alpha = 0.05 \)): \( 0.007223; \)
			\item P-value: \( 0.571420; \)
			\item Decision: \textbf{Fail to reject} \( H_0 \) (Uniform distribution).
		\end{itemize}
		
		\item \textbf{Beta(0.5,1) data (should reject)}: A sample of size \( n = 50 \) generated from Beta(0.5,1) yields:
		\begin{itemize}
			\item Test statistic: \( T_n = 0.015311; \)
			\item Critical value (\( \alpha = 0.05 \)): \( 0.007223; \)
			\item P-value: \( 0.000080; \)
			\item Decision: \textbf{Reject} \( H_0 \) (Uniform distribution).
		\end{itemize}
		\end{enumerate}
The WCPED-based uniformity test offers a valuable addition to the existing battery of goodness-of-fit tests. The test statistic is derived from the WCPED, providing a theoretically justified measure of discrepancy between distributions. The simulation results demonstrate that the test maintains good power properties across various alternatives and sample sizes, making it a reliable tool for practitioners.
\subsection{Application in image analysis}
Here, we present a comprehensive image analysis using the WCPED estimator to quantify dissimilarities between images and their resolution-reduced versions. Using $w(y)=y,$ the WCPED estimator, defined as:
\begin{equation}
	\widehat{D}_{wp}(X_1,X_2) = \int_{0}^{\infty}y ({F}_{1n}(y)-{F}_{2n}(y)){F}_{1n}(y)dy
	\label{eq:wcped},
\end{equation}
where ${F}_{1n}(\cdot)~~\text{and}~~{F}_{2n}(\cdot)$ defined in (\ref{eq2.10}) and provides a robust measure of divergence between two probability distributions. In this context, the distributions represent pixel intensity values from grayscale images, where pixel values range from $0$ (black) to $1$ (white). Two grayscale images, denoted as Image A (4059 $\times$ 3040 pixels) and Image B (2831 $\times$ 3537 pixels) are analyzed. Resolution-reduced versions have been generated at $75\%$, $50\%$, and $25\%$ of the original resolution, resulting in the following variants:
\begin{itemize}
	\item Image A: Original (12,339,360 pixels), A75 (6,940,320 pixels), A50 (3,085,600 pixels), A25 (771,400 pixels)
	\item Image B: Original (10,013,247 pixels), B75 (5,632,319 pixels), B50 (2,503,488 pixels), B25 (625,872 pixels)
\end{itemize}
The resolution-reduced versions of both images are presented in Figure \ref{fig:figure6}. The $\widehat{D}_{wp}(X_1,X_2)$ estimator is applied to compare same images at different resolutions (e.g., A vs A75, A vs A50, A vs A25) and different images at corresponding resolutions (e.g., A vs B, A75 vs B75).
The $\widehat{D}_{wp}(X_1,X_2)$  was estimated using the weight function w(y)=y. The results for the same-image comparisons are presented in Table \ref{tab:same}. As expected, the divergence between an image and itself is zero. For image A, the estimated magnitude of $\widehat{D}_{wp}(X_1,X_2)$ increases from 0.000253 at $75\%$ resolution to $0.004453$ at $50\%$ and $0.005419$ at $25\%$ resolution. This indicates an increasing dissimilarity as the resolution is progressively reduced. For image B, the corresponding magnitude of $\widehat{D}_{wp}(X_1,X_2)$  values are $0.000451$, $0.001757$, and $0.005286$, respectively. Although a small fluctuation is observed at $50\%$ resolution, the divergence becomes substantially larger at $25\%$, indicating greater dissimilarity at severe resolution reduction. The obtained  $|\widehat{D}_{wp}(X_1,X_2)|$ values remain close to zero even at $25\%$ resolution, indicating that the reduction in image resolution causes only a minor change in the underlying pixel-value distributions. This suggests that the essential distributional characteristics of the images are largely preserved despite substantial resolution reduction. Table \ref{tab:diff} reports the WCPED values between the two different images. The estimated $|\widehat{D}_{wp}(X_1,X_2)|$ values are 0.024861, 0.023897, 0.024354, and 0.028632 for the $100\%$, $75 \%$, $50\%$, and $25\%$ resolutions, respectively. These values are considerably larger than the corresponding same-image comparisons, demonstrating that the proposed WCPED estimator is able to distinguish between the two images based on their pixel-intensity distributions. Overall, the results demonstrate the usefulness of the WCPED estimator for measuring image dissimilarity and assessing the effect of resolution reduction.
\begin{figure}[htbp]
	\centering
	\includegraphics[width=0.8\textwidth]{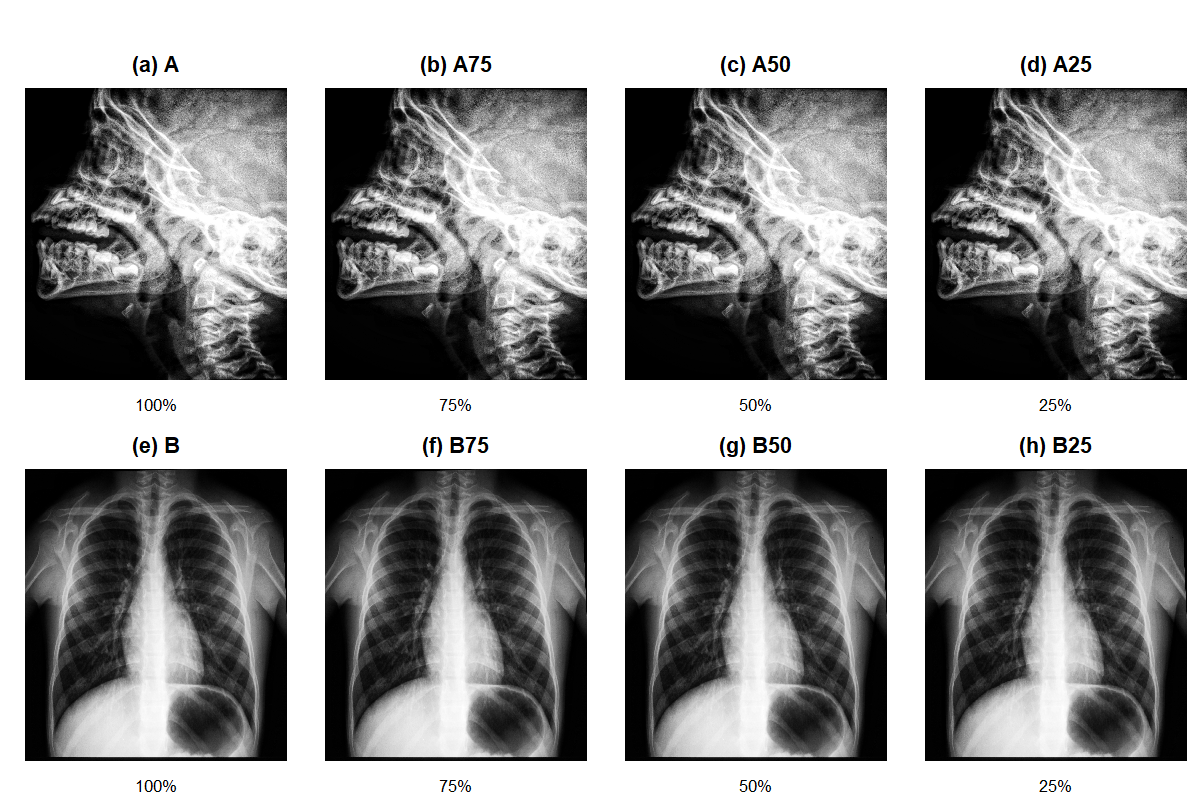}
	\caption{Original images and their corresponding resolution-reduced versions.}
	\label{fig:figure6}
\end{figure}

\begin{table}[H]
	\centering
	\caption{Empirical $\widehat{D}_{wp}(X_1,X_2)$  estimates for the same image comparisons.}
	\label{tab:same}
	\begin{tabular}{lccc}
		\toprule
		\textbf{Image 1} & \textbf{Image 2} & \textbf{Resolution} & $|\widehat{D}_{wp}|$ \\
		\midrule
		A & A & 100\% & 0.000000 \\
		A & A75 & 75\% & 0.000253 \\
		A & A50 & 50\% & 0.004453 \\
		A & A25 & 25\% & 0.005419 \\
		B & B & 100\% & 0.000000 \\
		B & B75 & 75\% & 0.000451 \\
		B & B50 & 50\% & 0.001757 \\
		B & B25 & 25\% & 0.005286 \\
		\bottomrule
	\end{tabular}
\end{table}
\begin{table}[H]
	\centering
	\caption{Empirical $\widehat{D}_{wp}(X_1,X_2)$  estimates for different image comparisons.}
	\label{tab:diff}
	\begin{tabular}{lcc}
		\toprule
		\textbf{Image 1} & \textbf{Image 2} & $|\widehat{D}_{wp}|$ \\
		\midrule
		A & B & 0.024861 \\
		A75 & B75 & 0.023897 \\
		A50 & B50 & 0.024354 \\
		A25 & B25 & 0.028632 \\
		\bottomrule
	\end{tabular}
\end{table}

\section{Dynamic weighted cumulative past extropy-inaccuracy and divergence}\label{sec4}
In reliability theory and survival analysis, the amount of information available about a random lifetime may change as time progresses. In particular, once an individual has survived up to a specified time point, the remaining lifetime provides additional information about the underlying distribution. This motivates the development of dynamic versions of information measures, where the measures are evaluated conditionally at a given time point. In this section, we introduce dynamic forms of the WCPEI and the WCPED measures. We also establish some useful properties of these measures.
\begin{definition}
	Suppose that $X_1$ and $X_2$ are two nonnegative continuous
	RVs with respective CDFs $F_1(\cdot)$ and $F_2(\cdot)$. For a fixed time point $t>0$, define the RVs by $X_{1,t}=(X_1-t\mid X_1\leq t), \quad X_{2,t}=(X_2-t\mid X_2\leq t).$ Then,
	the DWCPEI between $X_{1,t}$ and $X_{2,t}$ is defined as
	\begin{equation}\label{eq4.1}
	\xi I_{wp}(F_{1,t},F_{2,t};t)=-\frac{1}{2}\int_0^t w(y)
	\frac{F_1(y)F_2(y)}{F_1(t)F_2(t)}\,dy,
	\end{equation}
	provided that the integral exists.
\end{definition}
The above quantity measures the weighted past extropy-inaccuracy between
two distributions up to the time point $t$. When the two distributions coincide, the corresponding dynamic inaccuracy measure reduces to the dynamic WCPE of $X_1$, which is  defined in (\ref{eq1.4}).
Furthermore, the  DWCPEI can be written as:
\begin{equation}
\xi I_{wp}(F_{1,t},F_{2,t};t)
=
\frac{1}{F_1(t)F_2(t)}\left[I_{wp}(X_1,X_2)+
\frac{1}{2}\int_0^t w(y)F_1(y)F_2(y)\,dy\right]=
 h(t)(I_{wp}(X_1, X_2)+s_{w}(t)),
\end{equation}
where 
$h(t)=\frac{1}{F_1(t)F_2(t)}$ and $s_w(t)=
\frac{1}{2} \int_0^t w(y)F_1(y)F_2(y)\,dy.$
\begin{thm}\label{thm4.1}
	Let $X_{1,t}$ and $X_{2,t}$ be two RVs with CDFs $F_{1,t}(\cdot)$ and $F_{2,t}(\cdot)$, respectively. Then, $\xi I_{wp}(F_{1,t},F_{2,t};t)$ satisfies the following differential equation:
	\begin{eqnarray}\label{eq4.3}
	\begin{aligned}
		\frac{d}{dt}\xi I_{wp}(F_{1,t},F_{2,t};t)
		=-\bigl(r_{X_1}(t)+r_{X_2}(t)\bigr)
		\xi I_{wp}(F_{1,t},F_{2,t};t)-\frac{1}{2}w(t),
	\end{aligned}
	\end{eqnarray}
	where $r_{X_1}(t)=\frac{f_1(t)}{F_1(t)}\quad\text{and}\quad
	r_{X_2}(t)=\frac{f_2(t)}{F_2(t)}$ are the reversed hazard rates of $X_1$ and $X_2$, respectively.
\end{thm}
\begin{proof}
	Differentiating the DWCPEI measure in (\ref{eq4.1}), we directly get the expression in (\ref{eq4.3}).
\end{proof}
\begin{thm}
	The DWCPEI $\xi I_{wp}(F_{1,t},F_{2,t};t)$ is nondecreasing (nonincreasing) if and only if
	\begin{equation}
	\xi I_{wp}(F_{1,t},F_{2,t};t)
	\leq (\geq)\, \frac{-w(y)}{\bigl(r_{X_1}(t)+r_{X_2}(t)\bigr)}.
	\end{equation}
\end{thm}
\begin{proof}
The proof is simple, and thus it is omitted.
\end{proof}

Denote $\mu_{F_{1}}^{w}(t)=\int_{0}^{t}w(y)\frac{F_{1}(y)}{F_{1}(t)}dy$ and $\mu_{F_{2}}^{w}(t)=\int_{0}^{t}w(y)\frac{F_{2}(y)}{F_{2}(t)}dy$. The following theorem gives a bound of the proposed DWCPEI measure. 

\begin{thm}\label{th4.3}
	Consider two RVs $X_{1}$ and $X_{2}$ with respective CDFs $F_{1}(\cdot)$ and $F_{2}(\cdot)$. Then,
	\begin{eqnarray}
		\xi I_{wp}(F_{1,t},F_{2,t};t)\ge \max\{-\frac{1}{2}\mu_{F_{1}}^{w}(t), -\frac{1}{2}\mu_{F_{2}}^{w}(t)t\}.
	\end{eqnarray}
\end{thm}
\begin{proof}
	The proof of this theorem follows using the observation that $F_{1}(y)\le F_{1}(t)$ and $F_{2}(y)\le F_{2}(t),$ for $y\le t$.
\end{proof}

\begin{remark}
	The result in Theorem \ref{th4.3} reduces to Proposition $2$ of \cite{hashempour2026dynamic} if $w(y)=y.$
\end{remark}
	
	Next, we introduce the DWCPED measure.

\begin{definition}
	Suppose that $X_1$ and $X_2$ are two nonnegative continuous
	random variables with CDFs $F_1(\cdot)$ and $F_2(\cdot)$, respectively.
	Let $X_{1,t}$ and $X_{2,t}$ denote their corresponding past lifetime
	RVs at time $t$. Then, the DWCPED at time $t$ is defined as
	\begin{eqnarray}\label{eq4.5}
		\xi D_{wp}(F_{1,t},F_{2,t};t)
		&=&
		\frac{1}{2}
		\int_0^t
		w(y)
		\left(
		\frac{F_1(y)}{F_1(t)}
		-
		\frac{F_2(y)}{F_2(t)}
		\right)
		\frac{F_1(y)}{F_1(t)}
		\,dy.
	\end{eqnarray}
\end{definition}
Equation (\ref{eq4.5}) can be written in terms of the dynamic cumulative past extropy and the DCPEI as:
\begin{eqnarray}\label{eq4.6}
	\xi D_{wp}(F_{1,t},F_{2,t};t)= \xi I_{wp}(F_{1,t},F_{2,t};t)-J_{wp}(X_1;t).
\end{eqnarray}

\begin{thm}
	 The DWCPED
		$\xi D_{wp}(F_{1,t},F_{2,t};t)$ satisfies the following differential equation:
		\begin{equation}
		\frac{d}{dt}\xi D_{wp}(F_{1,t},F_{2,t};t)
		=
		-\bigl(r_{X_1}(t)+r_{X_2}(t)\bigr){\xi D}_{wp}(F_{1,t},F_{2,t};t)
		+\bigl(r_{X_1}(t)-r_{X_2}(t)\bigr)J_{wp}(X_1;t),
		\end{equation}
		where $J_{wp}(X_1;t)$ is defined in (\ref{eq1.4}).
	\end{thm}
	
	\begin{proof}
		Differentiating (\ref{eq4.6}) with respect to $t$, we obtain
		\begin{equation}\label{eq4.8}
		\frac{d}{dt}\xi D_{wp}(F_{1,t},F_{2,t};t)=\frac{d}{dt}\xi_{wp}(F_{1,t},F_{2,t};t)
		-\frac{d}{dt}J_{wp}(X_1;t).
		\end{equation}
		From Theorem \ref{thm4.1}, we have
		\begin{equation}\label{eq4.9}
		\frac{d}{dt}\xi I_{wp}(F_{1,t},F_{2,t};t)=-\frac{w(t)}{2}
		-\bigl(r_X(t)+r_Y(t)\bigr)\xi I_{wp}(F_{1,t},F_{2,t};t)
		\end{equation}
		and 
		\begin{equation}\label{eq4.10}
		\frac{d}{dt}J_{wp}(X_1;t)=-\frac{w(t)}{2}-2r_X(t)J_{wp}(X_1;t).
		\end{equation}
		Putting the expressions (\ref{eq4.9}) and (\ref{eq4.10}) in(\ref{eq4.8}), we get the result.
		\end{proof}
\begin{thm}
	The DWCPED
	$\xi D_{wp}(F_{1,t},F_{2,t};t)$ is nondecreasing (nonincreasing) if and only if
	\begin{equation}
	\xi D_{wp}(F_1,F_2;t)
	\leq (\geq)
	\frac{r_{X_1}(t)-r_{X_2}(t)}
	{r_{X_1}(t)+r_{X_2}(t)}
	J_{wp}(X_1;t),
	\end{equation}
	provided that $r_{X_1}(t)+r_{X_2}(t)>0$.
\end{thm}

\begin{proof}
	The proof is simple. Thus, it is omitted for brevity.
\end{proof}

Below, we use the concept of reversed hazard rate ordering (see \cite{shaked2007stochastic}) to establish an inequality for the DWCPED.
\begin{thm}
	If $X_1\ge_{rhr} X_2,$ then $\xi D_{wp}(F_{1,t},F_{2,t};t)\ge \xi D_{wp}(F_{2,t},F_{1,t};t).$
\end{thm}
\begin{proof}
Since $X_1\geq_{rhr}X_2$, then $r_{X_1}(t)\le r_{X_2}(t).$ 
By the definition of the reversed hazard rate order, the ratio $\frac{F_2(y)}{F_1(y)}$ is increasing in $y$. For $0\le y\leq t$,
		\begin{eqnarray}
		\frac{F_2(y)}{F_1(y)}\leq \frac{F_2(t)}{F_1(t)}\implies \frac{F_2(y)}{F_2(t)}\leq \frac{F_1(y)}{F_1(t)}&\implies&
		w(y)\left(\frac{F_2(y)}{F_2(t)}\right)^2 \leq  w(y)\left(\frac{F_1(y)}{F_1(t)}\right)^2 \nonumber \\&\implies& J_{wp}(X_2;t)\ge J_{wp}(X_1;t).
		\end{eqnarray}
		Now, from (\ref{eq4.6}) we have
		\begin{eqnarray}\xi D_{wp}(F_{1,t},F_{2,t};t)-\xi D_{wp}(F_{2,t},F_{1,t};t) &=& J_{wp}(X_2;t)- J_{wp}(X_1;t)\ge0\nonumber \\ \implies \xi D_{wp}(F_1,F_2;t)\ge \xi D_{wp}(F_2,F_1;t).
			\end{eqnarray}
			This completes the proof.
\end{proof}

Below, we have proposed kernel-based estimation procedures for estimating DWCPEI and DWCPED. It can be mentioned here that one can also estimate the same using the EDF-based technique, which has been illustrated in the case of estimating WCPEIR. The EDF-based technique is omitted for the sake of conciseness.

\subsection{Kernel estimation of the DWCPEI and DWCPED}

Let $Y_1,\cdots,Y_n$ and $Z_1,\cdots,Z_n$ be two
independent random samples from populations having
CDFs $F_1(\cdot)$ and $F_2(\cdot)$, respectively. Let $K(\cdot)$
denote a kernel function and let $b_n>0$ be a bandwidth satisfying
$b_n\to0$ as $n\to\infty$. The kernel estimators of the distribution functions $F_1(\cdot)$ and $F_2(\cdot)$ are
given by
\begin{equation}
\widehat{F}_{1n}(y)=\frac{1}{n}\sum_{i=1}^{n}K\left(\frac{y-Y_i}{b_n}\right)
~~~\text{and}~~~
\widehat{F}_{2n}(y)=\frac{1}{n}\sum_{i=1}^{n}K\left(\frac{y-Z_i}{b_n}\right).
\end{equation}
Consequently, the kernel estimators of the corresponding dynamic
past distribution functions are
\begin{equation}\label{eq4.14}
\widehat{F}_{1n,t}(y)=\frac{\widehat{F}_{1n}(y)}{\widehat{F}_{1n}(t)}
~~~\text{and}~~~
\widehat{F}_{2n,t}(y)=\frac{\widehat{F}_{2n}(y)}{\widehat{F}_{2n}(t)},
\qquad 0\leq x\leq t.
\end{equation}
\subsubsection{Kernel estimator of DWCPEI}
The DWCPEI measure is defined by
\[
\xi I_{wp}(F_{1,t},F_{2,t};t)=-\frac{1}{2}\int_0^tw(y)F_{1,t}(y)F_{2,t}(y)\,dy.
\]
Replacing $F_{1,t}(\cdot)$ and $F_{2,t}(\cdot)$ by their kernel estimates gives the
following plug-in estimator:
\begin{equation}
\widehat{\xi I}_{wp}(\widehat{F}_{1n,t},\widehat{F}_{2n,t};t)=-\frac{1}{2}\int_0^tw(y)\widehat{F}_{1n,t}(y)
\widehat{F}_{2n,t}(y)\,dy.
\end{equation}
Substituting the kernel estimates explicitly, we obtain
\begin{equation}
\begin{aligned}
	\widehat{\xi I}_{wp}(\widehat{F}_{1n,t},\widehat{F}_{2n,t};t)
	=-\frac{1}{2}\int_0^tw(y)
	\left[
	\frac{
		\displaystyle\sum_{i=1}^{n}
		K\left(\frac{y-Y_i}{b_n}\right)}
	{
		\displaystyle\sum_{i=1}^{n}
		K\left(\frac{t-Y_i}{b_n}\right)}
	\right]
	\left[
	\frac{
		\displaystyle\sum_{j=1}^{n}
		K\left(\frac{y-Z_j}{b_n}\right)}
	{
		\displaystyle\sum_{j=1}^{n}
		K\left(\frac{t-Z_j}{b_n}\right)}
	\right]
	\,dy.
\end{aligned}
\end{equation}
\subsubsection{Kernel estimator of DWCPED}
The DWCPED is defined as
\begin{equation*}
\widehat{\xi D}_{wp}(\widehat{F}_{1n,t},\widehat{F}_{2n,t};t)
=\frac{1}{2}\int_0^t w(y)\left\{F_{1,t}(y)-F_{2,t}(y)\right\}
F_{1,t}(y)\,dy.
\end{equation*}
Using the kernel estimates of the dynamic distribution functions which is in (\ref{eq4.14}), the corresponding plug-in estimator of DWCPED is obtained as
\begin{equation*}
\widehat{\xi D}_{wp}(\widehat{F}_{1n,t},\widehat{F}_{2n,t};t)
=\frac{1}{2}\int_0^tw(y)\left\{\widehat{F}_{1n,t}(y)
-\widehat{F}_{2n,t}(y)
\right\}
\widehat{F}_{1n,t}(y)\,dy.
\end{equation*}
Thus,
\begin{equation}
	\widehat{\xi D}_{wp}(\widehat{F}_{1n,t},\widehat{F}_{2n,t};t)
	=\frac{1}{2}\int_0^t
	w(y)\left[\frac{\widehat{F}_{1n}(y)}
	{\widehat{F}_{1n}(t)}
	-
	\frac{\widehat{F}_{2n}(y)}
	{\widehat{F}_{2n}(t)}
	\right]
	\frac{\widehat{F}_{1n}(y)}
	{\widehat{F}_{1n}(t)}
	\,dy.
\end{equation}
On substituting the kernel estimators, we get
\begin{equation}
	\widehat{\xi D}_{wp}(\widehat{F}_{1n,t},\widehat{F}_{2n,t};t)=\frac{1}{2}
	\int_0^t w(y)\left[ \frac{
		\displaystyle\sum_{i=1}^{n}
		K\left(\frac{y-Y_i}{b_n}\right)}
	{
		\displaystyle\sum_{i=1}^{n}
		K\left(\frac{t-Y_i}{b_n}\right)} -
		\frac{
		\displaystyle\sum_{j=1}^{n}
		K\left(\frac{y-Z_j}{b_n}\right)}
		{
		\displaystyle\sum_{j=1}^{n}
		K\left(\frac{t-Z_j}{b_n}\right)}\right] \frac{
		\displaystyle\sum_{i=1}^{n}
		K\left(\frac{y-Y_i}{b_n}\right)}
		{
		\displaystyle\sum_{i=1}^{n}
		K\left(\frac{t-Y_i}{b_n}\right)}
		\,dy.
\end{equation}
The above estimator is obtained by replacing the unknown distribution functions in the definition of DWPED by their kernel-based estimates.

\subsection{Simulation Study}
To assess the finite-sample performance of $\widehat{\xi I}_{wp}(\widehat{F}_{1n,t},\widehat{F}_{2n,t};t)$ and $\widehat{\xi D}_{wp}(\widehat{F}_{1n,t},\widehat{F}_{2n,t};t)$, a Monte Carlo simulation study is conducted. We consider independent samples from $X_1\sim\operatorname{Exp}(1)$ and $X_2\sim\operatorname{Exp}(2)$ with the weight function $w(y)=e^{-y}$. The sample sizes were taken as
$n=20,50,100,200,$ and $300$, while $t$ was varied over
$0.25,0.50,\ldots,2.00$. For each combination of $n$ and $t$, $10,000$ Monte Carlo replications were performed in $R$ software. The distribution functions are estimated using a Gaussian kernel estimator with Silverman's bandwidth, and numerical integration was carried out using the trapezoidal rule. The AB and MSE values are calculated by comparing the average estimated values with the corresponding theoretical values. The simulated result is reported in Table \ref{tab:7}.

The results show that the proposed estimators perform satisfactorily as the sample size increases. For $\widehat{\xi I}_{wp}(\widehat{F}_{1n,t},\widehat{F}_{2n,t};t)$ both the AB and MSE decrease consistently with increasing $n$ for all considered values of $t$. For example, at $t=1.5$, the AB decreases from $0.003522$ for $n=20$ to $0.002189$ for $n=300$, while the MSE decreases from $0.000298$ to $0.000027$. Similar observation for $\widehat{\xi D}_{wp}(F_{1,t},F_{2,t};t)$ is noticed. Overall, the results indicate that both kernel estimators exhibit improved accuracy and decreasing MSE with increasing sample size.
\begin{table}[htbp]
	\centering
	\caption{Simulation results for $\widehat{\xi D}_{wp}(\widehat{F}_{1n,t},\widehat{F}_{2n,t};t)$ and $\widehat{\xi I}_{wp}(\widehat{F}_{1n,t},\widehat{F}_{2n,t};t)$ estimates (Exp(1) vs Exp(2)).}
	\label{tab:7}
	\begin{tabular}{ccccrrrrrr}
		\toprule
		\multirow{2}{*}{$t$} & \multirow{2}{*}{$n$} & \multicolumn{4}{c}{$\widehat{\xi }_{wp}(\widehat{F}_{1n,t},\widehat{F}_{2n,t};t)$} & \multicolumn{4}{c}{$\widehat{\xi I}_{wp}(\widehat{F}_{1n,t},\widehat{F}_{2n,t};t)$} \\
		\cmidrule(lr){3-6} \cmidrule(lr){7-10}
		& & True & Estimate & AB & MSE & True & Estimate & AB & MSE \\
		\midrule
		0.25 & 20 & -0.001146 & 0.007508 & 0.008655 & 0.000127 & -0.038013 & -0.052458 & 0.014446 & 0.000246 \\
		0.25 & 50 & -0.001146 & 0.006930 & 0.008076 & 0.000086 & -0.038013 & -0.050824 & 0.012811 & 0.000178 \\
		0.25 & 100 & -0.001146 & 0.006219 & 0.007365 & 0.000065 & -0.038013 & -0.048889 & 0.010877 & 0.000125 \\
		0.25 & 200 & -0.001146 & 0.005376 & 0.006522 & 0.000047 & -0.038013 & -0.046920 & 0.008907 & 0.000083 \\
		0.25 & 300 & -0.001146 & 0.004788 & 0.005934 & 0.000038 & -0.038013 & -0.045808 & 0.007796 & 0.000063 \\
		\midrule
		0.50 & 20 & -0.004015 & 0.004338 & 0.008354 & 0.000159 & -0.069594 & -0.076533 & 0.006939 & 0.000122 \\
		0.50 & 50 & -0.004015 & 0.002525 & 0.006541 & 0.000077 & -0.069594 & -0.075272 & 0.005678 & 0.000060 \\
		0.50 & 100 & -0.004015 & 0.001053 & 0.005068 & 0.000042 & -0.069594 & -0.073830 & 0.004237 & 0.000032 \\
		0.50 & 200 & -0.004015 & -0.000282 & 0.003733 & 0.000022 & -0.069594 & -0.072591 & 0.002997 & 0.000016 \\
		0.50 & 300 & -0.004015 & -0.000997 & 0.003019 & 0.000014 & -0.069594 & -0.071938 & 0.002344 & 0.000010 \\
		\midrule
		1.25 & 20 & -0.016488 & -0.014096 & 0.002392 & 0.000164 & -0.135404 & -0.132638 & 0.002766 & 0.000236 \\
		1.25 & 50 & -0.016488 & -0.015180 & 0.001307 & 0.000067 & -0.135404 & -0.132803 & 0.002600 & 0.000101 \\
		1.25 & 100 & -0.016488 & -0.015579 & 0.000909 & 0.000035 & -0.135404 & -0.133150 & 0.002254 & 0.000054 \\
		1.25 & 200 & -0.016488 & -0.015965 & 0.000523 & 0.000018 & -0.135404 & -0.133428 & 0.001975 & 0.000031 \\
		1.25 & 300 & -0.016488 & -0.016062 & 0.000425 & 0.000012 & -0.135404 & -0.133555 & 0.001848 & 0.000022 \\
		\midrule
		1.50 & 20 & -0.020559 & -0.018043 & 0.002517 & 0.000192 & -0.150038 & -0.146516 & 0.003522 & 0.000298 \\
		1.50 & 50 & -0.020559 & -0.019464 & 0.001095 & 0.000076 & -0.150038 & -0.146737 & 0.003301 & 0.000131 \\
		1.50 & 100 & -0.020559 & -0.019874 & 0.000685 & 0.000040 & -0.150038 & -0.147121 & 0.002916 & 0.000072 \\
		1.50 & 200 & -0.020559 & -0.020221 & 0.000338 & 0.000020 & -0.150038 & -0.147598 & 0.002440 & 0.000039 \\
		1.50 & 300 & -0.020559 & -0.020245 & 0.000314 & 0.000014 & -0.150038 & -0.147849 & 0.002189 & 0.000027 \\
		\midrule
		1.75 & 20 & -0.024233 & -0.021936 & 0.002296 & 0.000214 & -0.161937 & -0.157571 & 0.004366 & 0.000365 \\
		1.75 & 50 & -0.024233 & -0.023082 & 0.001151 & 0.000084 & -0.161937 & -0.158313 & 0.003624 & 0.000158 \\
		1.75 & 100 & -0.024233 & -0.023549 & 0.000684 & 0.000043 & -0.161937 & -0.158754 & 0.003183 & 0.000084 \\
		1.75 & 200 & -0.024233 & -0.023912 & 0.000321 & 0.000022 & -0.161937 & -0.159365 & 0.002572 & 0.000046 \\
		1.75 & 300 & -0.024233 & -0.023959 & 0.000274 & 0.000015 & -0.161937 & -0.159501 & 0.002436 & 0.000032 \\
		\midrule
		2.00 & 20 & -0.027438 & -0.024797 & 0.002642 & 0.000231 & -0.171549 & -0.166875 & 0.004674 & 0.000412 \\
		2.00 & 50 & -0.027438 & -0.026128 & 0.001310 & 0.000090 & -0.171549 & -0.167543 & 0.004006 & 0.000186 \\
		2.00 & 100 & -0.027438 & -0.026770 & 0.000669 & 0.000045 & -0.171549 & -0.167941 & 0.003608 & 0.000098 \\
		2.00 & 200 & -0.027438 & -0.027048 & 0.000391 & 0.000023 & -0.171549 & -0.168659 & 0.002890 & 0.000052 \\
		2.00 & 300 & -0.027438 & -0.027147 & 0.000292 & 0.000015 & -0.171549 & -0.168970 & 0.002579 & 0.000036 \\
		\bottomrule
	\end{tabular}
\end{table}
\section{Conclusion}\label{sec5}
In this paper, we have developed a unified framework of weighted cumulative past extropy-based information measures for comparing two lifetime distributions. By introducing a nonnegative weight function in the cumulative past extropy framework, we proposed the WCPEI and WCPED measures. The weighting mechanism provides additional flexibility in emphasizing different regions of the lifetime distribution and allows several existing measures to be recovered as special cases. We established a number of theoretical properties of the proposed measures and investigated their relationships. These results demonstrate that the proposed measures are capable of capturing meaningful differences between lifetime distributions from the perspective of past lifetime uncertainty. In addition, suitable nonparametric estimators were developed, and their finite sample behavior was examined through simulation studies for different distributions and sample sizes. The simulation results indicate that the proposed estimators provide satisfactory performance, with their AB and  generally decreasing as the sample size increases. The practical applicability of the proposed methodology was illustrated through two different applications. First, an extropy based goodness-of-fit test for the uniform distribution was constructed using the proposed divergence measure. Its empirical power was compared with several commonly used uniformity tests, including the Kolmogorov--Smirnov, Anderson--Darling, Cramér--von Mises, Zamanzade, and another extropy-based test. The results show that the proposed test is particularly effective for detecting certain centered alternatives, especially those represented by the $C_i$ family. Although the proposed test does not uniformly dominate all competing procedures across every alternative, its strong performance for the centered alternatives highlights the usefulness of the extropy based approach in situations where departures from uniformity are concentrated around the center of the distribution.
Second, the proposed measures were applied to image analysis to quantify differences between pixel-value distributions at different image resolutions. The results demonstrate that the proposed measure can effectively assess the information discrepancy caused by image resolution reduction. In particular, the relatively small values of the measure for moderate reductions in resolution indicate that the overall distribution of pixel intensities can remain largely preserved, illustrating the potential of the proposed framework as a distribution-based tool for image comparison.

To extend the proposed methodology to situations in which the distributional behavior changes with time, we further introduced the DWCPEI and DWCPED measures. These dynamic measures are constructed by conditioning on the elapsed lifetime up to a specified time point and therefore provide a time dependent characterization of the discrepancy between past lifetime distributions. Their main mathematical properties were established, and corresponding estimation procedures were developed. Simulation studies under different distributional settings and sample sizes were conducted to investigate the finite-sample performance of the proposed estimators.

\section*{Acknowledgements}  
The authors gratefully acknowledge the financial support received during this research. Bighneswar Sahoo thanks the University Grants Commission (UGC), India, for providing financial assistance under Award No. 231620132839. Suchandan Kayal gratefully acknowledges the research funding provided by the Anusandhan National Research Foundation (ANRF) under Grant No. CRG/2023/000157. The authors also sincerely thank the Department of Mathematics, National Institute of Technology Rourkela, India, for providing the necessary research facilities and support to carry out this work.

 \section*{Conflicts of interest} All the authors declare no conflict of interest. 
  
    \section*{Data Availability Statement} The data sets used in this manuscript are openly available. 
  
	\bibliography{ref}

@article{kerridge1961inaccuracy,
	title={Inaccuracy and inference},
	author={Kerridge, David F},
	journal={Journal of the Royal Statistical Society. Series B (Methodological)},
	pages={184--194},
	year={1961},
	publisher={JSTOR}
}

@article{shannon1948mathematical,
	title={A mathematical theory of communication},
	author={Shannon, Claude Elwood},
	journal={The Bell System Technical Journal},
	volume={27},
	number={3},
	pages={379--423},
	year={1948},
	publisher={Nokia Bell Labs}
}

@book{shaked2007stochastic,
	title={Stochastic orders},
	author={Shaked, Moshe and Shanthikumar, J George},
	year={2007},
	publisher={Springer}
}

@article{kullback1951information,
	title={On information and sufficiency},
	author={Kullback, Solomon and Leibler, Richard A},
	journal={The Annals of Mathematical Statistics},
	volume={22},
	number={1},
	pages={79--86},
	year={1951},
	publisher={JSTOR}
}

@article{lad2015extropy,
	title={Extropy: Complementary dual of entropy},
	author={Lad, Frank and Sanfilippo, Giuseppe and Agro, Gianna},
	journal={Statistical Science},
	volume={30},
	number={1},
	pages={40--58},
	year={2015},
	publisher={JSTOR}
}

@article{qiu2017extropy,
	title={The extropy of order statistics and record values},
	author={Qiu, Guoxin},
	journal={Statistics \& Probability Letters},
	volume={120},
	pages={52--60},
	year={2017},
	publisher={Elsevier}
}

@article{kundu2023cumulative,
	title={On cumulative residual (past) extropy of extreme order statistics},
	author={Kundu, Chanchal},
	journal={Communications in Statistics-Theory and Methods},
	volume={52},
	number={16},
	pages={5848--5865},
	year={2023},
	publisher={Taylor \& Francis}
}

@article{hashempour2024new,
	title={A new measure of inaccuracy for record statistics based on extropy},
	author={Hashempour, Majid and Mohammadi, Morteza},
	journal={Probability in the Engineering and Informational Sciences},
	volume={38},
	number={1},
	pages={207--225},
	year={2024},
	publisher={Cambridge University Press}
}

@article{saranya2025inaccuracy,
	title={Inaccuracy and divergence measures based on survival extropy with applications in testing and image analysis},
	author={Saranya, P and Sunoj, SM},
	journal={Japanese Journal of Statistics and Data Science},
	volume={8},
	number={2},
	pages={921--947},
	year={2025},
	publisher={Springer}
}

@article{saranya2026relative,
	title={Relative and divergence measures based on extropy: dynamic forms and properties},
	author={Saranya, P and Sunoj, SM},
	journal={Probability in the Engineering and Informational Sciences},
	pages={1--23},
	year={2026},
	publisher={Cambridge University Press}
}

@article{saranya2024relative,
	title={On relative cumulative extropy, its residual (past) measures and their applications in estimation and testing},
	author={Saranya, P and Sunoj, SM},
	journal={Journal of the Indian Society for Probability and Statistics},
	volume={25},
	number={1},
	pages={199--225},
	year={2024},
	publisher={Springer}
}

@article{hashempour2024dynamic,
	title={On dynamic cumulative past inaccuracy measure based on extropy},
	author={Hashempour, Majid and Mohammadi, Morteza},
	journal={Communications in Statistics-Theory and Methods},
	volume={53},
	number={4},
	pages={1294--1311},
	year={2024},
	publisher={Taylor \& Francis}
}

@article{an1933sulla,
	title={Sulla determinazione empirica di una legge didistribuzione},
	author={Kolmogorov, An},
	journal={Giorn Dell'inst Ital Degli Att},
	volume={4},
	pages={89--91},
	year={1933}
}

@article{hashempour2026dynamic,
	title={Dynamic perspectives on weighted past cumulative extropy inaccuracy: characterizations and applications},
	author={Hashempour, Majid and Mohammadi, Morteza},
	journal={Statistics},
	pages={1--29},
	year={2026},
	publisher={Taylor \& Francis}
}

@article{anderson1954test,
	title={A test of goodness of fit},
	author={Anderson, Theodore W and Darling, Donald A},
	journal={Journal of the American statistical association},
	volume={49},
	number={268},
	pages={765--769},
	year={1954},
	publisher={Taylor \& Francis}
}

@article{cramer1928composition,
	title={On the composition of elementary errors: First paper: Mathematical deductions},
	author={Cram{\'e}r, Harald},
	journal={Scandinavian Actuarial Journal},
	volume={1928},
	number={1},
	pages={13--74},
	year={1928},
	publisher={Taylor \& Francis}
}

@article{von1945wahrscheinlichkeitsrechnung,
	title={Wahrscheinlichkeitsrechnung und ihre Anwendung in der Statistik und theoretischen Physik},
	author={Von Mises, Richard},
	journal={Deuticke, Leipzig and Vienna},
	year={1931}
}

@article{zamanzade2015testing,
	title={Testing uniformity based on new entropy estimators},
	author={Zamanzade, Ehsan},
	journal={Journal of Statistical Computation and Simulation},
	volume={85},
	number={16},
	pages={3191--3205},
	year={2015},
	publisher={Taylor \& Francis}
}

@article{qiu2018extropy,
	title={Extropy estimators with applications in testing uniformity},
	author={Qiu, Guoxin and Jia, Kai},
	journal={Journal of Nonparametric Statistics},
	volume={30},
	number={1},
	pages={182--196},
	year={2018},
	publisher={Taylor \& Francis}

}
	
\end{document}